\RequirePackage{fix-cm}
\documentclass[smallextended]{svjour3}       
\smartqed  
\usepackage{graphicx}
 \usepackage{mathptmx}      
\usepackage[numbers]{natbib}

\usepackage[utf8]{inputenc}

\usepackage{amsmath}
\usepackage{amsfonts}
\usepackage{amssymb}

\usepackage{placeins}
\usepackage{algorithm}
\usepackage{algorithmic}
\usepackage{booktabs}

\usepackage{color}
\usepackage[dvipsnames]{xcolor}
\usepackage[colorlinks=true,linkcolor=blue,citecolor=blue,pdfborder={0 0 0}]{hyperref}

\usepackage{bbm}
\usepackage{dsfont}
\usepackage{bm}
\usepackage[normalem]{ulem}
\usepackage{graphicx}
\usepackage{enumerate}
\usepackage{paralist}
 
\renewcommand{\emptyset}{\varnothing}

\newcommand{\rset}{\mathbb{R}}

\newcommand{\ind}{\mathbbm{1}}
\newcommand{\un}{\ind}
\newcommand{\indinf}{\boldsymbol{\iota}}
\newcommand{\dd}{\{1,\ldots,d\}}
\newcommand{\bzero}{\boldsymbol{0}}

\newcommand{\PP}{\operatorname{\mathbb{P}}}
\newcommand{\PE}{\operatorname{\mathbb{E}}}
\newcommand{\Var}{\operatorname{\mathbb{V}\mathrm{ar}}}
\newcommand{\Cov}{\operatorname{\mathbb{C}\mathrm{ov}}}
\newcommand{\mb}{\bm}
\newcommand{\ud}{\,\mathrm{d}}

\newcommand{\fidis}[1][{}]{\emph{fidis}}

\def\keywords{\vspace{.5em} {\textit{Keywords}:\, \relax%
  }}

\newcommand{\reals}{\rset}
\newcommand{\eps}{\varepsilon}

\newcommand{\wto}{\rightsquigarrow}

\newcommand{\oh}{\mathrm{o}}
\newcommand{\Oh}{\mathrm{O}}

\renewcommand{\le}{\leqslant}
\renewcommand{\ge}{\geqslant}

\newtheorem{condition}{Condition}
\newtheorem{cor}[theorem]{Corollary}

 \journalname{Extremes}
\begin{document}

\title{Identifying groups of variables with the potential of being large simultaneously
}


\author{Maël Chiapino \and Anne Sabourin
\and 
Johan Segers
}


\institute{Maël Chiapino \and Anne Sabourin \at
  LTCI, Télécom ParisTech, Université Paris-Saclay \\ 46, rue Barrault, 75013 Paris, France.\\
              \email{mael.chiapino@telecom-paristech.fr\,,\,anne.sabourin@telecom-paristech.fr}           
           \and
          Johan Segers \at 
          Universit\'{e} catholique de Louvain, Institut de Statistique, Biostatistique et Sciences Actuarielles, \\Voie du Roman Pays~20, B-1348 Louvain-la-Neuve, Belgium.\\
          \email{johan.segers@uclouvain.be}
}

\date{Received: date / Accepted: date}

\maketitle

\begin{abstract}
  Identifying groups of variables that may be large simultaneously amounts to finding out which joint tail dependence coefficients of a multivariate distribution are positive. The asymptotic distribution of a vector of nonparametric, rank-based estimators of these coefficients justifies a stopping criterion in an algorithm that searches the collection of all possible groups of variables in a systematic way, from smaller groups to larger ones. The issue that the tolerance level in the stopping criterion should depend on the size of the groups is circumvented by the use of a conditional tail dependence coefficient. Alternatively, such stopping criteria can be based on limit distributions of rank-based estimators of the coefficient of tail dependence, quantifying the speed of decay of joint survival functions. Numerical experiments indicate that the algorithm's effectiveness for detecting tail-dependent groups of variables is highest when paired with a criterion based on a Hill-type estimator of the coefficient of tail dependence.

  \keywords{multivariate extremes ;  asymptotic dependence ; statistical tests ; high dimensional data}
  \subclass{MSC 62G32 \and MSC 62H15 \and MSC 62H05 \and MSC 62H20}
\end{abstract}
%
%

\section{Introduction}
  
A question that often arises when monitoring several variables is which groups of variables are prone to be large simultaneously. In food risk management, for instance, the variables under consideration may be the concentrations of different contaminants in blood samples of consumers. In environmental applications, one may be interested in several physical variables such as wind speed and precipitation recorded at several locations, with the purpose of setting off a regional warning when several of these variables exceed a high threshold. In the context of semi-supervised anomaly detection, when the training sample is mostly made of normal instances, identifying the groups of variables which are likely to be large together allows to label certain new instances as abnormal.

The latter use case is the motivation behind the DAMEX algorithm \cite{goixsparse,goix2017sparse}. In a regular variation framework, identifying those groups among $d$ variables that may be large simultaneously amounts to identifying the support of the exponent measure. The algorithm returns the list of groups of features $\alpha\subset\{1,\ldots, d\}$ such that the mass of the empirical exponent measure on certain cones exceeds a user-defined threshold. However, when the empirical version of the exponent measure is scattered over a large number of such cones, the DAMEX algorithm does not discover a clear-cut structure. \citet{chiapinofeature} encounter this difficulty for extreme streamflow data recorded at several locations of the French river system.

To overcome this issue, the same authors come up with the CLEF (CLustering Extreme Features) algorithm. Instead of partitioning the sample space, CLEF considers nested regions corresponding to increasing subsets of components. A group of variables is enlarged until there is no longer enough evidence that all features in it may be large together. In this respect, CLEF resembles the Apriori algorithm \cite{agrawal1994fast}, which is a data-mining tool for discovering maximal sets of items among $d$ available items that are frequently  bought together by consumers. Apriori considers increasing itemsets that are made to grow until their frequency falls below a user-defined threshold. In CLEF, the stopping criterion concerns the relative frequency of simultaneous occurrences of large values of all components in a considered subset compared to the frequency of simultaneous occurrences of larges values of all but one component in this subset. \citet{chiapinofeature} find the method to work well on real and simulated data but do not investigate the asymptotic properties of the statistic underlying the stopping criterion.

Our contributions are three-fold. First, we investigate the asymptotic behavior of the statistic underlying CLEF. In this way, the informal stopping criterion can be turned into a proper hypothesis test with controllable level. A second issue concerns the specification of the null hypothesis in the CLEF stopping criterion. Originally, a certain conditional tail dependence coefficient, $\kappa_\alpha$, related to a given group of variables $\alpha \subset \{1, \ldots, d\}$ is supposed to be above a strictly positive, user-defined and therefore somewhat arbitrary threshold. We propose instead to base the stopping criterion on the hypothesis that a multivariate version of the coefficient of \citet{ledford1996statistics} and \citet{ramos2009new} is equal to one. The test is based on the limit distributions of multivariate extensions of nonparametric estimators in \citet{peng1999estimation} and \citet{draisma2001tail, draisma2004bivariate}. Third, we conduct a numerical experiment to compare the finite-sample performance of the DAMEX algorithm and the CLEF algorithm with the various stopping criteria. We find that overall, the multivariate extension of the Hill-type estimator in \cite{draisma2004bivariate} yields the most reliable procedure to detect maximal groups of asymptotically dependent variables.

Section~\ref{sec:taildep_background} casts the problem in the language of regular variation and introduces the tail dependence coefficients upon which the CLEF stopping criteria will be based. Necessary background on empirical tail dependence functions and processes is reviewed in Section~\ref{sec:etdf}, including a new result for the empirical joint tail function. In Section~\ref{sec:test-kappa}, we derive the asymptotic distribution of the statistic used in CLEF and turn the heuristic stopping criterion implemented in \cite{chiapinofeature} into a statistical test with asymptotically controllable level. Two alternative tests based on the asymptotic distributions of estimators of the Ledford--Tawn--Ramos coefficient of tail dependence are constructed in Sections~\ref{sec:mult-extens-peng} and~\ref{sec:hill}. We report the results of our simulation experiments in  Section~\ref{sec:simu-study}. Section~\ref{sec:conclusion} concludes. Proofs are gathered in Appendix~\ref{sec:appendix} while the pseudo-code for the CLEF algorithm and variations is provided in Appendix~\ref{sec:appendix-CLEF}.





\section{Regular variation and tail dependence coefficients}
\label{sec:taildep_background}

Bold letters denote vectors and binary operations between vectors are understood componentwise. 
The indicator function of a set $A$ is denoted by $\un_A$. For $t\in \rset\cup\{\infty\}$, we let $\mb t_\alpha$ denote the constant vector of $(\rset \cup \{\infty\})^\alpha$ with all coordinates equal to $t$. In the special case $\alpha = \{1, \ldots, d\}$, the index $\alpha$ is usually omitted for brevity when clear from the context: for instance, $\bzero = \bzero_{\{1,\ldots,d\}} = (0, \ldots, 0) \in \reals^d$.

Let ${\mb X} = (X_1, \ldots, X_d)$ be a random vector in $\reals^d$ with cumulative distribution function $F$, whose margins $F_1, \ldots, F_d$ are continuous. We assume that the transformed vector ${\mb V} = (V_1, \ldots, V_d)$ with $V_j = 1/\{1 - F_j(X_j)\}$ for all $j \in \{1,\ldots,d\}$ is regularly varying on the cone $[0,\infty]^d \setminus \{\bzero\}$ with (nonzero) limit or exponent measure $\mu$. This means that $\mu$ is finite on Borel sets of $[0,\infty]^d \setminus \{\bzero\}$ bounded away from the origin and that 
\begin{equation}
  \label{eq:reg-var}
  \lim_{t \to \infty} t \PP [ {\mb V} \in  t A] = \mu(A),
\end{equation}
for all Borel sets $A \subset [0,\infty]^d \setminus \{\bzero\}$ such that $\bzero \notin\partial A$ and $\mu(\partial A)=0$. The measure $\mu$ is homogeneous, i.e., $\mu(s \,\cdot\,) = s^{-1} \mu(\,\cdot\,)$ for all $0 < s < \infty$, and therefore assigns no mass to hyperplanes parallel to the coordinate axes. As a consequence, \eqref{eq:reg-var} applies to finite and infinite rectangles that are bounded away from the origin and whose sides are parallel to the coordinate axes. The measure $\mu$ characterizes the extremal dependence structure of ${\mb X}$. The reader is referred to \citet{resnick:2007, resnick2013extreme} for an introduction to regular variation.  

Let $\varnothing \ne \alpha \subset \dd$. Particular instances of \eqref{eq:reg-var} include the extremal coefficient $\lambda_\alpha$ \cite{schlather2003dependence} and the joint tail coefficient 
$\rho_\alpha$: 
\begin{align}
\label{eq:lambda_alpha}
  \lambda_\alpha
  &=
  \lim_{t \to \infty}  t \PP [ \exists j \in \alpha : V_j >  t ]
  =
  \mu ( \{ \mb{u} \in [0, \infty)^d \mid \exists j \in \alpha : u_j > 1 \} ),
  \\
  \label{eq:rho_alpha}
  \rho_\alpha
  &=
  \lim_{t \to \infty}  t \PP [ \forall j \in \alpha : V_j >  t ]
  =
  \mu ( \{ \mb{u} \in [0, \infty)^d \mid \forall j \in \alpha : u_j > 1 \} ).
\end{align}
In the bivariate case, $|\alpha|=2$, and with our choice of Pareto margins, we have $\rho_\alpha = \lim_{t\to \infty} \PP(V_{\alpha_1}>t \mid V_{\alpha_2}>t)$, the upper tail dependence coefficient denoted by $\chi$ in \cite{coles1999dependence}.

Our general objective is to propose statistically sound procedures to recover maximal subgroups $\alpha$ of components that are likely to be concomitantly large. Our aim can thus be phrased as recovering the maximal subsets $\alpha\subset\dd$ such that $\rho_\alpha>0$.

Since $\rho_\alpha \le \rho_\beta$ as soon as $\alpha \supset \beta$, any positive tolerance level with which we would like to compare an estimate of $\rho_\alpha$ should depend on $\alpha$ and in particular be decreasing as a function of the cardinality $|\alpha|$. To circumvent this issue, \citet{chiapinofeature} consider for $\alpha$ such that $|\alpha| \ge 2$ the conditional tail dependence coefficient
\begin{equation}
\label{eq:kappa}
  \kappa_\alpha =
  \lim_{t\to \infty} \PP\left[
  \forall j \in \alpha : V_j > t 
  \; \Big|\; 
  \textstyle\sum_{j\in\alpha} \ind{\{V_j > t\}}  \ge |\alpha| -1    
  \right],
\end{equation}
which is the limiting conditional probability that all variables in $\alpha$ exceed a large threshold given that all but at most one already do. In contrast to $\rho_\alpha$, the coefficient $\kappa_\alpha$ has no particular reason to decrease as a function of $|\alpha|$. Note that $\rho_\alpha = \mu(\Gamma_\alpha)$ while $\kappa_\alpha = \mu(\Gamma_\alpha) / \mu(\Delta_\alpha) = \rho_\alpha / \mu(\Delta_\alpha)$ where $\Gamma_\alpha = \{ \mb{x} \in [0, \infty)^d \mid \forall j \in \alpha : x_j > 1 \}$ and $\Delta_\alpha = \{\mb x \in [0, \infty)^d \mid \textstyle\sum_{j \in \alpha} \ind_{\{x_j\ge 1\}}\ge |\alpha|-1\}$, provided $\lvert \alpha \rvert \ge 2$. If $\mu(\Delta_\alpha) = 0$, then $\mu(\Gamma_\beta) = 0$ for all $\beta \subset \alpha$ with $|\beta| = |\alpha| - 1$; in that case, we define $\kappa_\alpha = 0$.

In the CLEF algorithm \citep{chiapinofeature}, the criterion to decide whether $\rho_\alpha>0$ or not is that $\widehat\kappa_\alpha \ge C$, where $C$ is a user-defined tolerance level, $\widehat\kappa_\alpha = \widehat\mu(\Gamma_\alpha)/\widehat\mu(\Delta_\alpha)$, and $\widehat\mu$ is the empirical exponent measure in \eqref{eq:mun} below. The level $C$ can be chosen independently of $\alpha$. Still, its choice is somewhat arbitrary, and in particular, the user has no control of false positives. In Section~\ref{sec:test-kappa}, we will provide the asymptotic distribution of $\widehat\kappa_\alpha$ and propose a test statistic with a guaranteed asymptotic level.

If $\rho_\alpha=0$ (or $\kappa_\alpha=0$), the limiting distributions of the statistics $\sqrt{k} (\widehat{\rho}_\alpha - \rho_\alpha)$ and $\sqrt{k} (\widehat{\kappa}_\alpha - \kappa_\alpha)$ are degenerate at zero.  We therefore have no control on the asymptotic levels of tests based on those statistics under $H_0 : \kappa_0 = 0$. This is why will have to define a CLEF stopping criterion in terms of a test of $H_0 : \kappa_\alpha \ge \kappa_{\min}$ versus $H_1 : \kappa_\alpha < \kappa_{\min}$ instead, in terms of a user-defined level $\kappa_{\min} > 0$. The choice of $\kappa_{\min}$ is somewhat arbitrary; in the simulation experiments (Section~\ref{sec:simu-study}), we choose $\kappa_{\min} = 0.08$. 

In Sections~\ref{sec:mult-extens-peng} and~\ref{sec:hill}, we consider alternative CLEF stopping criteria based on estimators of the coefficient of tail dependence $\eta_\alpha\in (0,1]$. For bivariate distributions, the coefficient has been introduced by \citet{ledford1996statistics} and extended by \citet{ramos2009new} in order to model situations in between asymptotic dependence ($\rho_{\{1,2\}} > 0$) and full independence of $X_1$ and $X_2$. \citet{HaanZhou11residual} and \citet{eastoe2012subasymptotic} proposed and studied a multivariate extension of $\eta_\alpha$ for $|\alpha| \ge 3$. The model assumption is that there exist $\eta_\alpha\in(0,1]$ and a slowly varying function $\mathcal{L}_\alpha$ such that
\begin{equation}
  \label{eq:taildep-multi}
  \PP[ \forall j \in \alpha : V_j > t ] = t^{-1/\eta_\alpha} \mathcal{L}_\alpha(t).  
\end{equation}

Suppose that the limit $\rho_\alpha$ in \eqref{eq:rho_alpha} exists and that \eqref{eq:taildep-multi} holds. Then $\rho_\alpha > 0$ implies $\eta_\alpha = 1$. The converse is true as well, provided $\liminf_{t\to\infty}\mathcal{L}_\alpha(t) > 0$. Modulo this side condition, which we will take for granted, the null hypothesis $\rho_\alpha>0$ corresponds to the simple hypothesis $\eta_\alpha = 1$. 

We will test the null hypothesis $\eta_\alpha = 1$ via multivariate extensions of nonparametric estimators of $\eta_\alpha$ in \citet{peng1999estimation} and \citet{draisma2004bivariate}. The null limit of the test statistic is non-degenerate, so that the asymptotic level of the test can be controlled, with no need to introduce an additional tolerance parameter $\kappa_{\min}$. The estimators that we will study are related to the Pickands estimator and the Hill estimator for the extreme value index of $T_\alpha = \min_{j\in\alpha} V_j$, respectively. The maximum likelihood estimator, also considered in~\cite{draisma2004bivariate}, is less suitable to our context due to its relative computational complexity, since the test is destined to be performed on a large number of subsets of $\{1,\ldots,d\}$. See also the review \cite{bacro2013measuring} and the references therein.

\begin{remark}
\label{rm:relationship-goix}
The DAMEX algorithm \citep{goix2017sparse} is designed to recover the family $\mathcal{M}$ of non-empty subsets $\alpha$ of $\dd$ with the property that
\[
  \mu \Bigl(
    \Bigl\{ \mb x \in [0, \infty)^d \;\Big|\;
    \|\mb x\|_\infty \ge 1 ;\; 
    \forall j \in \alpha\,, x_j > 0\; \text{ and } \forall j
    \notin\alpha, \,x_j = 0 \Bigr\}
  \Bigr)
  > 0.
\]
In contrast, our focus is on $\mathbb{M} = \{ \alpha \mid \rho_\alpha > 0\} = \{ \alpha \mid \kappa_\alpha > 0 \}$. Still,  the maximal elements of $\mathbb{M}$ for the inclusion order are also the maximal elements of $\mathcal{M}$ \citep[Lemma~1]{chiapinofeature}. The two problems of finding the maximal elements of $\mathbb{M}$ or $\mathcal{M}$ are thus equivalent.
\end{remark}

\section{Empirical tail dependence functions and processes}
\label{sec:etdf}

To find the asymptotic distribution of nonparametric estimators of the
various dependence coefficients, we rely on empirical tail
processes. Let the random vector $\mb{X} \sim F$ be as in
Section~\ref{sec:taildep_background}; in particular, assume regular
variation as in \eqref{eq:reg-var} with exponent measure $\mu$. Let
$\Lambda$ be the push-forward measure of $\mu$ on
$[0, \infty]^d \setminus \{ \mb{\infty} \}$ induced by the
transformation $\mb{x} \mapsto 1/\mb{x} = (1/x_1, \ldots, 1/x_d)$,
i.e.,
$\Lambda(\,\cdot\,) = \mu( \{ \mb{x} \in [0, \infty]^d \setminus \{
\bm{0} \} \mid 1/\mb{x} \in \, \cdot \, \} )$.

For $\varnothing \ne \alpha \subset \dd$, consider the stable tail dependence function $\ell_\alpha : [0, \infty)^\alpha \to [0, \infty)$ 
and the joint tail dependence function $r_\alpha : [0, \infty]^\alpha \setminus \{ \mb{\infty}_\alpha \} \to [0, \infty)$ given by
\begin{align}
\nonumber
  \ell_\alpha(\mb x) 
  &= \lim_{t \to 0} t^{-1} \PP [ \exists j \in \alpha : F_j(X_j) > 1 -  t x_j]
  = \Lambda( \{ \mb{y} \mid \exists j \in \alpha : y_j < x_j \} ),
  \\
\label{eq:r_alpha}
  r_\alpha(\mb x) 
  &= \lim_{t \to 0} t^{-1} \PP [ \forall j \in \alpha : F_j(X_j) > 1 - t x_j] 
  = \Lambda( \{ \mb{y} \mid \forall j \in \alpha : y_j < x_j \} ).
\end{align}
From \eqref{eq:lambda_alpha} and \eqref{eq:rho_alpha}, clearly $\lambda_\alpha = \ell_\alpha(\mb{1}_\alpha)$ and $\rho_\alpha = r_\alpha(\mb{1}_\alpha)$. For brevity, we write $\ell = \ell_{\dd}$ and $r = r_{\dd}$. Note that $\ell_\alpha( \bm{x} ) = \ell( \bm{x} \bm{e}_\alpha )$ for $\bm{x} \in [0, \infty)^\alpha$, where $\mb e_\alpha \in \{0, 1\}^d$ has components $\mb e_{\alpha, j} = \un_{\alpha}(j)$. Similarly, $r_\alpha( \bm{x} ) = r( \bm{x} \bm{\iota}_\alpha )$ for $\bm{x} \in [0, \infty]^\alpha \setminus \{ \bm{\infty}_\alpha \}$, where $\indinf_\alpha \in \{1,\infty\}^d$ denotes the vector such that $\indinf_{\alpha,j} = 1 $ if $j \in \alpha$ and $\indinf_{\alpha,j} = +\infty $ otherwise. By the inclusion--exclusion formula, for $\bm{x} \in [0, \infty)^\alpha$, 
writing $\bm{x}_\beta = (x_j)_{j \in \beta}$, we have
\begin{align}
  \label{eq:ell_beta2r_alpha}
  r_\alpha(\mb x) 
  &= 
  \sum_{\varnothing \ne \beta \subset \alpha}
  (-1)^{|\beta|+1} \ell_\beta (\mb x_\beta), &
  \ell_\alpha(\mb x) 
  &= 
  \sum_{\emptyset \neq \beta \subset \alpha}
  (-1)^{|\beta|+1} r_\beta (\mb x_\beta).
\end{align}
Let ${\mb X}_i = (X_{i,1}, \ldots, X_{i,d})$, for $i \in \{1,\ldots,n\}$, be an independent random sample from $F$, having continuous margins and satisfying \eqref{eq:reg-var}. Let $k = k(n)\to \infty$ as $n\to\infty$, while $k(n) = \oh(n)$. Following for instance \cite{einmahl2012m, goix2017sparse, qi1997almost}, we rely on ranks to obtain an approximately Pareto-distributed sample $\widehat {\mb V}_i = (\widehat{V}_{i,1}, \ldots, \widehat{V}_{i,d})$. Let $\widehat F_j(x) = n^{-1} \sum_{i=1}^n \ind_{\{X_{i,j} < x\}}$ be the (left-continuous) empirical distribution function of component $j \in \{1,\ldots,d\}$ and put $\widehat V_{i,j} = 1/\{1 - \widehat F_j(X_{i,j})\} = n/(n + 1 - R_{i,j})$, where $R_{i,j}$ 
is the rank of $X_{i,j}$ among $X_{1,j}, \ldots, X_{n,j}$. The empirical counterparts to $\mu$ and $\Lambda$ are
\begin{align}
  \label{eq:mun}
    \widehat{\mu}(\,\cdot\,) 
    &= 
    \frac{1}{k}\sum_{i=1}^n \delta_{ (k/n)\widehat{\mb V}_{i} }(\,\cdot\,), &
    \widehat{\Lambda}(\,\cdot\,)
    &=
    \frac{1}{k}\sum_{i=1}^n \delta_{ (n/k)/\widehat{\mb V}_{i} }(\,\cdot\,),
\end{align}
respectively, with $\delta_a$ the Dirac measure at the point $a$. Replacing $\Lambda$ by $\widehat{\Lambda}$ in the definition of $\ell_\alpha$ and $r_\alpha$ produces the empirical tail dependence function
\begin{align*}
  \widehat{\ell}_\alpha( \bm{x} )
  &= k^{-1} \textstyle\sum_{i=1}^n \un\{\exists j \in \alpha : n + 1 - R_{i,j} \le \lfloor k x_j \rfloor \} \\
  &= k^{-1} \textstyle\sum_{i=1}^n \un\{\exists j \in \alpha : X_{i,j} \ge X_{(n-\lfloor k x_j \rfloor +1),j} \} 
\end{align*}
and the empirical joint tail function
\begin{align}
\label{eq:r_alpha:estim}
  \widehat{r}_\alpha( \bm{x} )
  &= k^{-1} \textstyle\sum_{i=1}^n \un\{\forall j \in \alpha : n + 1 - R_{i,j} \le \lfloor k x_j \rfloor \} \\
\nonumber  
  &= k^{-1} \textstyle\sum_{i=1}^n \un\{\forall j \in \alpha : X_{i,j} \ge X_{(n-\lfloor k x_j \rfloor +1),j} \},    
\end{align}
where $X_{(1),j} \le \ldots \le X_{(n),j}$ are the ascending order statistics of $X_{1,j}, \ldots, X_{n,j}$ and $\lfloor \,\cdot\, \rfloor$ is the floor function. The identities \eqref{eq:ell_beta2r_alpha} hold for $\widehat{\ell}_\alpha$ and $\widehat{r}_\alpha$ as well.

\citet[Theorem~4.6]{einmahl2012m} find the weak limit of the empirical process $\sqrt{k} ( \widehat{\ell} - \ell )$ on $[0, T]^d$ for any $T > 0$. We leverage their theorem to show a similar result for $\sqrt{k}( \widehat{r}_\alpha - r_\alpha )$, jointly in $\alpha$. The following conditions stem from the cited article.
\begin{condition}[Uniform tail convergence]
\label{as:bias}
There exists $\gamma > 0$ such that, uniformly in $\mb x \in [0,1]^d$ with $\sum_{j=1}^d x_j = 1$, we have 
\[
  t^{-1}\PP[ \exists j = 1, \ldots, d : \, F_j(X_j) > tx_j ] - \ell(\mb x) 
  = \Oh(t^{\gamma}), \qquad t \to \infty.
\]
\end{condition}

\begin{condition}[Moderate $k$]
\label{as:small-k}
The sequence $k = k(n)$ satisfies $k = \oh(n^{2\gamma/(1+2\gamma)})$ as $n \to \infty$, with $\gamma > 0$ as in Condition~\ref{as:bias}.
\end{condition}

\begin{condition}[Smoothness]
\label{as:partialDeriv}
For all $j \in\{1,\ldots,d\}$, the partial derivative $\partial_j \ell = \partial \ell / \partial x_j$ exists and is continuous on the set $\{ \mb x \in [0,\infty)^d \mid x_j > 0\}$.
\end{condition}

Since $\ell$ is convex, it is continuously differentiable Lebesgue almost everywhere \citep[Theorem~25.5]{rockafellar:1970}. Condition~\ref{as:partialDeriv} is satisfied for many popular max-stable models (logistic, asymmetric logistic, Brown--Resnick) but fails for max-linear models. Under Condition~\ref{as:partialDeriv}, the partial derivative $\partial_j r_\alpha = \partial r_\alpha / \partial x_j$ ($j \in \alpha$) exists and is continuous on $\{ \mb{x} \in [0, \infty)^\alpha \mid x_j > 0 \}$ and satisfies $\partial_j r_\alpha( \mb{x} ) = \sum_{\beta : j \in \beta \subset \alpha} (-1)^{|\beta|+1} \partial_j \ell_\beta(\mb{x}_\beta)$, where $\mb{x}_\beta = (x_s)_{s \in \beta}$.

\citet{einmahl1997poisson} and \citet{einmahl2012m} consider a centered Gaussian process $W$ indexed by the Borel sets of $[0,\infty]^d\setminus\{\mb{\infty} \}$ bounded away from $\mb{\infty}$ with covariance function 
\begin{equation}
\label{eq:covW}
  \PE[W(A) \, W(B)] = \Lambda(A \cap B ).
\end{equation}
Note that $W(\emptyset) = 0$ almost surely. For $\emptyset\neq \alpha\subset\{1,\ldots, d\}$ and $\bm{x} \in [0, \infty)^\alpha$, write
\[ 
  W_\alpha( \bm{x} ) = W(\{ \mb{y} \in [0, \infty]^d \mid \forall j \in \alpha : y_j < x_j \}). 
\]
We consider weak convergence as in \cite{van2000asymptotic, van1996weak}; notation $\wto$. We work in the metric space $\ell^\infty(S)$ of bounded, real functions $f$ on an arbitrary set $S$, the metric being the one induced by the supremum norm, $\| f \|_\infty = \sup_{x \in S} \lvert f(x) \rvert$; the double use of the symbol $\ell$ should not give rise to any confusion. The proof of the following proposition and of other results in the paper is deferred to Appendix~\ref{sec:appendix}.

\begin{proposition}
\label{prop:rnx}
Let ${\mb X}_i = (X_{i,1}, \ldots, X_{i,d})$, for $i \in \{1,\ldots,n\}$, be an independent random sample from $F$, having continuous margins and satisfying \eqref{eq:reg-var}. Let $k = k(n)\to \infty$ as $n\to\infty$, while $k(n) = \oh(n)$. If Conditions~\ref{as:bias},~\ref{as:small-k} and~\ref{as:partialDeriv} hold, then, for $T > 0$, in the product space $\prod_{\varnothing \ne \alpha \subset \dd} \ell^\infty( [0, T]^\alpha )$, we have, as $n \to \infty$, the weak convergence
\begin{equation}
\label{eq:Z_alpha(x)}
  \sqrt k\left\{\widehat r_\alpha(\mb x) - r_\alpha(\mb x)\right\}
  \wto
  W_\alpha(\mb x) - \sum_{j\in\alpha}
  \partial_j r_{\alpha}(\mb x) \, W_{\{j\}}(x_j)
  =
  Z_\alpha( \mb{x} ).
\end{equation}
\end{proposition}  
%
%

\section{Estimating the conditional tail dependence coefficient}
\label{sec:test-kappa}

This section investigates the asymptotic distribution of the empirical conditional dependence coefficient $\widehat{\kappa}_\alpha$ based on the empirical exponent measure $\widehat{\mu}$. This is achieved by re-writing  $\widehat\kappa_\alpha$  as a function of the empirical joint tail coefficients  $\widehat\rho_\alpha$, the distribution of which follows from Proposition~\ref{prop:rnx}. We also propose consistent estimators of the asymptotic variance of $\widehat{\kappa}_\alpha$. Combining the two yields a test for the null hypothesis $\kappa_\alpha \ge \kappa_{\min}$ where $\kappa_{\min}\in(0,1)$ is a tolerance level fixed by the user, to be seen as the minimal limiting conditional probability that all components in a random vector exceed a threshold, given that all of them but at most one already do.  

Let $\varnothing \ne \alpha \subset \dd$ and recall the sets $\Gamma_\alpha = \{ \mb{x} \in [0, \infty)^d \mid \forall j \in \alpha : x_j > 1 \}$ and, provided $\alpha$ has at least two elements, $\Delta_\alpha = \{\mb x \in [0, \infty)^d \mid \textstyle\sum_{j \in \alpha} \ind_{\{x_j\ge 1\}}\ge |\alpha|-1\}$. Write $\alpha \setminus j = \alpha \setminus \{j\}$ for $j \in \alpha$. Since $\Delta_\alpha$ is the disjoint union of the sets $\Gamma_{\alpha \setminus j} \setminus \Gamma_\alpha$ and $\Gamma_\alpha$, where $j \in \alpha$, we find, for every Borel measure $\nu$, the equality
\begin{equation}
  \label{eq:B-Gamma}
  \nu( \Delta_\alpha ) 
  = \sum_{j \in \alpha} \nu(\Gamma_{\alpha \setminus j}) - (|\alpha| -1) \, \nu(\Gamma_\alpha).
\end{equation}
Recall $\rho_\alpha = \mu(\Gamma_\alpha)$ and $\kappa_\alpha = \mu(\Gamma_\alpha) / \mu(\Delta_\alpha)$ in \eqref{eq:kappa}. By \eqref{eq:B-Gamma} applied to $\nu = \mu$, we have
\begin{equation}
  \label{eq:rewriteKappa}
  \kappa_\alpha 
  = 
  \frac%
    {\rho_\alpha}%
    {\sum_{j \in \alpha} \rho_{\alpha \setminus j} - (|\alpha| - 1)\rho_{\alpha}}.
\end{equation}

Recall the joint tail function $r_\alpha$ and its nonparametric estimator $\widehat{r}_\alpha$ in \eqref{eq:r_alpha} and \eqref{eq:r_alpha:estim}, respectively. Since $\rho_\alpha = r_\alpha( \mb{1}_\alpha )$, we define the estimators $\widehat{\rho}_\alpha = \widehat{\mu}( \Gamma_\alpha ) = \widehat{r}_\alpha( \bm{1}_\alpha )$ and, provided $\lvert \alpha \rvert \ge 2$,
\[
  \widehat{\kappa}_\alpha
  =
  \frac{\widehat{\mu}( \Gamma_\alpha )}{\widehat{\mu}( \Delta_\alpha )}
  =
  \frac%
    {\widehat{\rho}_\alpha}%
    {\sum_{j \in \alpha} \widehat{\rho}_{\alpha \setminus j} - (|\alpha| - 1) \widehat{\rho}_{\alpha}}.  
\]
The asymptotic distribution of the vector of empirical joint tail coefficients follows immediately from Proposition~\ref{prop:rnx}. Write $\dot{\rho}_{\alpha, j} = \partial_j r_{\alpha}(\bm{1}_\alpha)$. 

\begin{cor}
\label{prop:asymptotic-hatRho} 
In the setting of Proposition~\ref{prop:rnx}, we have, jointly in $\varnothing \ne \alpha \subset \dd$, the weak convergence
\begin{equation}
\label{eq:G_alpha}
  \sqrt{k_n} \left( \widehat{\rho}_{\alpha} - \rho_\alpha \right)
  \wto Z_\alpha( \mb{1}_\alpha ) = G_\alpha,
  \qquad n \to \infty.
\end{equation}
The limit distribution is centered Gaussian with covariance matrix
\begin{equation}
\label{eq:covG}
  \PE[ G_{\alpha} G_{\alpha'} ]
  = \rho_{\alpha\cup\alpha'} -
  \sum_{j \in \alpha} \dot{\rho}_{j,\alpha}
  \rho_{\alpha' \cup \{j\}} - \sum_{j' \in \alpha'}
  \dot{\rho}_{j',\alpha'} \rho_{\alpha \cup \{j'\}}
  +\sum_{j \in \alpha} \sum_{j' \in \alpha'}
  \dot{\rho}_{j,\alpha} \,  \dot{\rho}_{j',\alpha'} \, \rho_{\{j, j'\}}.
\end{equation}
\end{cor}

The asymptotic distribution of $\widehat{\kappa}_\alpha$ follows from the one of $(\widehat{\rho}_\beta)_\beta$ via the delta method. The asymptotic variance involves the partial derivative $\partial_j \kappa_\alpha = \partial \kappa_\alpha / \partial x_j$ of the function
\begin{equation}
  \label{eq:kappax}
  \kappa_\alpha(\mb x) 
   = \frac{r_\alpha(\mb x)}
   { \sum_{j \in\alpha} r_{\alpha \setminus j}(\mb x_{\alpha \setminus j}) - (\lvert\alpha\rvert -1) r_\alpha(\mb x)}
\end{equation}
for $\mb x \in [0,\infty)^\alpha$. Note that $\kappa_\alpha(\mb 1_\alpha) = \kappa_\alpha$. Write $\dot{\kappa}_{j,\alpha} = \partial_j \kappa_\alpha( \mb{1}_\alpha )$.
  
\begin{proposition}
\label{theo:asymptot-kappa}
In the setting of Corollary~\ref{prop:asymptotic-hatRho}, we have, as $n \to \infty$ and jointly in $\alpha\subset\{1,\ldots,d\}$ such that $|\alpha| \ge 2$ and $\mu(\Delta_\alpha) > 0$, the weak convergence
  \begin{equation}
  \label{eq:kappalimit}
    \sqrt{k}\left( \widehat\kappa_\alpha - \kappa_\alpha\right)
    \wto
    \mu(\Delta_\alpha)^{-2}
    \left\{
      \left( \textstyle\sum_{j\in\alpha}\rho_{\alpha\setminus j} \right) G_\alpha
      - \rho_\alpha \textstyle\sum_{j\in\alpha} G_{\alpha\setminus j}
    \right\}.
  \end{equation}
For a fixed such $\alpha$, the limit distribution is $\mathcal{N}(0, \sigma_{\kappa,\alpha}^2)$ with
  \begin{multline}
    \label{eq:sigmakappa2}
    \sigma^2_{\kappa, \alpha} = \big(1 - \kappa_\alpha)\kappa_\alpha
    \left\{ \mu(\Delta_\alpha)^{-1} - \textstyle\sum_{j\in\alpha}\dot{\kappa}_{j,\alpha} \right\}
    + \sum_{i\in\alpha}\sum_{j\in\alpha} \dot{\kappa}_{i,\alpha}\dot{\kappa}_{j,\alpha} \rho_{\{i,j\}} \\
    + \kappa_{\alpha} \sum_{ j \in\alpha}\dot{\kappa}_{j,\alpha}
    \left\{ 1 - \mu(\Delta_\alpha)^{-1} \rho_{\alpha\setminus j} \right\}.
  \end{multline}
\end{proposition}

Following \cite{peng1999estimation}, the asymptotic variance $\sigma^2_{\kappa,\alpha}$ in \eqref{eq:sigmakappa2} can be estimated consistently by estimating the partial derivatives $\dot{\kappa}_{i,\alpha}$ via finite differencing applied to the empirical version of $\kappa_{\alpha}(\mb{x})$ in \eqref{eq:kappax} obtained by replacing $r_\alpha$ and $r_{\alpha \setminus j}$ by $\widehat{r}_\alpha$ and $\widehat{r}_{\alpha \setminus j}$, respectively:
\[
  \widehat \kappa_{\alpha}(\mb x) 
  = 
  \frac%
    {\sum_{i=1}^n \un\{\forall j \in\alpha : X_{i,j } \ge X_{(n - \lfloor kx_j \rfloor +1 ), j} \} }%
    {\sum_{i=1}^n \un\{\exists m\in\alpha: \forall j \in\alpha \setminus m : X_{i,j } \ge  X_{(n - \lfloor kx_j \rfloor +1 ), j} \}} 
\]
Define
\begin{equation}
\label{eq:dotkappan}
  \dot{\kappa}_{j, \alpha, n} = \frac{1}{2 k^{-1/4}} 
  \left\{ \widehat \kappa_{\alpha}(\mb
  1_\alpha + k^{-1/4}\mb e_j) - \widehat \kappa_{\alpha}(\mb 1_\alpha - k^{-1/4}\mb e_j ) \right\},
\end{equation}
with $\mb e_j$ the canonical unit vector of $\reals^\alpha$ pointing in direction $j \in \alpha$, and put
\begin{multline}
\label{eq:Hatvar-kappa}
  \widehat{\sigma}^2_{\kappa, \alpha} = \big(1 -
  \widehat\kappa_\alpha)\widehat\kappa_\alpha 
  \left\{ \widehat\mu(\Delta_\alpha)^{-1} - \textstyle\sum_{j\in\alpha}\dot{\kappa}_{j,\alpha, n} \right\}
  + 
  \sum_{i,j\in\alpha}\dot{\kappa}_{i,\alpha, n}\dot{\kappa}_{j,\alpha,n} \widehat\rho_{\{i,j\}}  \\
  + 
  \widehat\kappa_{\alpha} \sum_{ j \in\alpha} \dot{\kappa}_{j,\alpha,n} 
  \left\{ 1 - \widehat\mu(\Delta_\alpha)^{-1} \widehat\rho_{\alpha \setminus j} \right\}.
\end{multline}

\begin{proposition}\label{prop:estim-sigma-kappa}
Under the conditions of Proposition~\ref{theo:asymptot-kappa}, we have
  $\widehat{\sigma}^2_{\kappa,\alpha} = \sigma^2_{\kappa,\alpha} + \oh_{\PP}(1)$ as $n \to \infty$,
  so that $\sqrt{k}(\widehat \kappa_\alpha - \kappa_\alpha) / \widehat{\sigma}_{\kappa,\alpha} \wto \mathcal{N}(0,1)$, provided $\sigma^2_{\kappa,\alpha} > 0$.
\end{proposition}

The proof relies on the weak convergence of the empirical process $\sqrt{k} \{ \widehat{\kappa}_{\alpha}(\,\cdot\,) - \kappa_\alpha(\,\cdot\,) \}$ on $[0,T]^\alpha$ for any $T>0$. This property follows in turn from Proposition~\ref{prop:rnx} and the functional delta method.

We consider a tolerance level $\kappa_{\min} \in (0, 1)$ under which the tail dependence between components $j\in\alpha$ is deemed negligible compared to the one between components $j\in\beta\subsetneq \alpha$.  In other words, we aim at testing $H_0: \kappa_\alpha\ge \kappa_{\min}$. Since $\kappa_\alpha = \rho_\alpha / \mu(\Delta_\alpha)$, the null hypothesis is that $\rho_\alpha$ is greater than some level depending on $\alpha$. Let $0<\delta<1$ be a (small) probability, and consider the test
\begin{equation}
  \label{eq:testKappa-ge}
  \tau_{\alpha,n} 
  = 
  \un\left\{
    \widehat\kappa_\alpha < \kappa_{\min} + q_\delta k^{-1/2} \widehat{\sigma}_{\kappa,\alpha}
  \right\} 
\end{equation}
where $q_\delta$ is the $\delta$-quantile of the standard normal distribution. By Proposition~\ref{prop:estim-sigma-kappa}, if $\sigma_{\kappa,\alpha} > 0$, the test in \eqref{eq:testKappa-ge} has asymptotic level $\delta$ for $H_0$ against $H_1: \kappa_{\alpha}< \kappa_{\min}$.

If $\rho_\alpha = 0$, then, in Proposition~\ref{prop:asymptotic-hatRho}, we have $\sqrt{k}( \widehat{\rho}_\alpha - \rho_\alpha ) = \oh_{\PP}(1)$ as $n \to \infty$: indeed, on the one hand, we have $\sqrt{k}( \widehat{\rho}_\alpha - \rho_\alpha ) = \sqrt{k} \widehat{\rho}_\alpha \ge 0$, and on the other hand, its limit distribution is centered Gaussian. Likewise, we have $\sqrt{k}( \widehat{\kappa}_\alpha - \kappa_\alpha ) = \oh_{\PP}(1)$ as $n \to \infty$ in Proposition~\ref{theo:asymptot-kappa} if $\kappa_\alpha = 0$. As a consequence, under the simple hypothesis $H_0 : \rho_\alpha = 0$, the asymptotic level of a test based on the asymptotic distribution of $\sqrt{k}(\widehat{\rho}_\alpha - \rho_\alpha)$ or $\sqrt{k}(\widehat{\kappa}_\alpha - \kappa_\alpha)$ cannot be controlled. This is why the test in \eqref{eq:testKappa-ge} concerns the null hypothesis $H_0 : \kappa_\alpha \ge \kappa_{\min}$ for some $\kappa_{\min} > 0$ instead. Alternatively, we propose tests based on estimators of the coefficient of tail dependence $\eta_\alpha$ in \eqref{eq:taildep-multi}. In Sections~\ref{sec:mult-extens-peng} and~\ref{sec:hill}, we consider two such estimators, extending the ones of \citet{peng1999estimation} and \citet{draisma2004bivariate}, respectively, to the multivariate setting.

\section{Coefficient of tail dependence: Peng's estimator}
\label{sec:mult-extens-peng}

For bivariate distributions, Peng's \citep{peng1999estimation} estimator of the coefficient of tail dependence $\eta = \eta_{\{1,2\}}$ is based on the property that the curve $t\mapsto (\log t, \log\PP[V_1>t, V_2>t])$ has an affine asymptote with slope $- 1/\eta$. A similar idea motivates Pickands' \citep{pickands1975statistical} estimator for the extreme value index. Estimating the ordinate of the curve at $t = n/k$ and $t = n/(2k)$ allows to estimate that slope. Under a second-order regular variation condition, \citet{peng1999estimation} shows that his estimator is asymptotically normal, both if $\eta = 1$ and if $ \eta < 1$. In the former case, the asymptotic variance depends on the tail dependence function and its partial derivatives, which are unknown but may be estimated consistently, thus leading to tests whose asymptotic levels can be controlled.

Let $\alpha \subset \dd$ have at least two elements. Recall the empirical joint tail function $\widehat{r}_\alpha$ in \eqref{eq:r_alpha:estim}. We define the multivariate extension of Peng's \citep{peng1999estimation} estimator of $\eta_\alpha$ in \eqref{eq:taildep-multi} as 
\begin{equation}
\label{eq:multi-peng}
  \widehat\eta_\alpha^P 
  = \log(2) /
  \log \{ \widehat r_{\alpha}(\mb 2_\alpha) / \widehat r_{\alpha}(\mb 1_\alpha) \}. 
\end{equation}
The asymptotic normality of $\widehat\eta_\alpha^P$ follows from Proposition~\ref{prop:rnx} and
the delta method.

\begin{proposition}
\label{thm:multi-peng}
In the setting of Proposition~\ref{prop:rnx}, we have, as $n \to \infty$ and jointly in $\alpha \subset \dd$ such that $\lvert \alpha \rvert \ge 2$ and $\rho_\alpha > 0$, the weak convergence
  \[
    \sqrt{k} (\widehat\eta_{\alpha}^P - 1) \wto 
    \frac{-1}{2\rho_\alpha \log 2} \left\{ Z_\alpha(\mb 2_\alpha) - 2 Z_\alpha(\mb 1_\alpha) \right\}. 
  \]
  The right-hand side is a $\mathcal{N}(0, \sigma_{\alpha,P}^2)$ random variable with variance
  \begin{multline}
    \label{eq:var-hateta}
    \sigma_{\alpha,P}^2 = \frac{1}{2(\rho_\alpha\log 2)^2}
    \biggl[
      \rho_\alpha - 4 \rho_\alpha^2 +  2 \sum_{j\in\alpha} \dot{\rho}_{j,\alpha} 
      r_\alpha(\mb 2_\alpha \wedge \indinf_{j}) \\
      + \sum_{j\in\alpha} \sum_{j'\in\alpha}
      \dot{\rho}_{j,\alpha}\dot{\rho}_{j',\alpha} \left\{ 3 \rho_{\{j, j'\}} - 2
      r_{\{j, j' \}}(2,1) \right\} 
    \biggr], 
  \end{multline}
  where $\rho_{\{j, j'\}} = r_{\{j, j'\}}(2, 1) = 1$ if $j = j'$ and
  where $\indinf_j \in \{1,\infty\}^\alpha$ is the vector which all coordinates equal to $1$ except for the $j$-th one which equals $\infty$, so that     $(\mb 2_\alpha \wedge \indinf_j)_m = 1$ if  $m \in \alpha\setminus j$ and $(\mb 2_\alpha \wedge \indinf_j)_m = 2$ if $m = j$.
\end{proposition}

By extending the proof of \cite[Theorem~2.1]{peng1999estimation}, it is also possible to obtain asymptotic normality of $\widehat{\eta}_\alpha^P$ in the case $\rho_\alpha=0$ and $\eta_\alpha < 1$ in~\eqref{eq:taildep-multi}. This would require a   multivariate extension of the second-order regular variation condition in \cite{peng1999estimation} in the style of Condition~\ref{as:draisma04} below. For the application as a stopping criterion in the CLEF algorithm, we are only interested in the asymptotic distribution of $\widehat\eta_\alpha^P$ under the hypothesis $\rho_\alpha>0$, so we do not pursue this idea any further. 

As in Proposition~\ref{theo:asymptot-kappa}, the asymptotic variance $\sigma_{\alpha,P}^2$ in \eqref{eq:var-hateta} involves unknown quantities, all of which we can estimate consistently. For $\alpha\subset\{1,\ldots,d\}$ and $j\in\alpha$, define
\begin{equation}
\label{eq:dotrho_j,alpha,n}
  \dot{\rho}_{j,\alpha,n} 
  = 
  \frac{1}{2 k^{-1/4}} 
  \left\{ 
    \widehat r_\alpha(\mb 1_\alpha + k^{-1/4}\mb e_j) 
    - 
    \widehat r_\alpha(\mb 1_\alpha - k^{-1/4}\mb e_j) 
  \right\},
\end{equation}
where $\mb{e}_j$ is the canonical unit vector in $\reals^\alpha$ pointing in dimension $j$. Define
\begin{multline}
\label{eq:Hatvar-hateta}
  \widehat \sigma^2_{\alpha,P} 
  = \frac{1}{2(\widehat \rho_\alpha\log 2)^2}
  \biggl[ 
    \widehat \rho_\alpha 
    + 
    \sum_{j\in\alpha} 
      \dot{\rho}_{j,\alpha, n} 
      \{ 
	- 4 \widehat \rho_\alpha 
	+ 2 \widehat r_\alpha(\mb 2_\alpha \wedge \indinf_{j}) 
      \} 
    \\
    + 
    \sum_{j\in\alpha} \sum_{j'\in\alpha} 
      \dot{\rho}_{j,\alpha,n} \dot{\rho}_{j',\alpha,n} 
      \left\{ 3 \widehat{\rho}_{\{j, j'\}} - 2 \widehat{r}_{\{j, j'\}}(2,1) \right\} 
  \biggr].
\end{multline}

\begin{proposition}
\label{prop:estimVarTerms}
In the setting of Proposition~\ref{prop:rnx}, we have $\widehat{\sigma}_{\alpha, P}^2 = \sigma_{\alpha, P}^2 + \oh_{\PP}(1)$ as $n \to \infty$, where $\alpha \subset \dd$ is such that $\lvert \alpha \rvert \ge 2$ and $\rho_\alpha > 0$. If $\sigma_{\alpha, P}^2 > 0$, then $\sqrt{k} (\widehat \eta_\alpha^P - 1) / \widehat{\sigma}_{\alpha, P} \wto \mathcal{N}(0,1)$ as $n \to \infty$.
\end{proposition}

The proof parallels the one of Proposition~\ref{prop:estim-sigma-kappa} and is omitted for brevity. The main step is to verify that $\dot{\rho}_{j,\alpha, n} = \dot{\rho}_{j,\alpha} + \oh_{\PP}(1)$ as $n \to \infty$, which follows from Proposition~\ref{prop:rnx}. 

To test the hypothesis $H_0 : \rho_\alpha > 0$ at significance level $\delta \in (0, 1)$, we propose
\begin{equation}
  \label{eq:test-eta}
  \tau_{\alpha,\eta^P,n} 
  = 
  \un \left\{
    \widehat\eta_\alpha^P 
    < 1 - q_{1-\delta} k^{-1/2} \widehat{\sigma}_{\alpha,P}
  \right\},
\end{equation}
where $q_{1-\delta}$ is the $(1-\delta)$-quantile of the standard normal distribution. In the setting of Proposition~\ref{prop:estimVarTerms}, the test in \eqref{eq:test-eta} has asymptotic level $\delta$ for $H_0$ against $H_1: \eta_{\alpha} < 1$.

\section{Coefficient of tail dependence: Hill estimator}
\label{sec:hill}

The coefficient of tail dependence $\eta_\alpha$ in \eqref{eq:taildep-multi} is the tail index of the random variable $T_\alpha = \min_{j \in \alpha} V_j$: the function $t \mapsto \PP[T_\alpha > t]$ is regularly varying at infinity with index $-1/\eta_\alpha$. A tractable alternative to Peng's estimator for $\eta_\alpha$ is a Hill-type estimator as in \citet{draisma2001tail,draisma2004bivariate}. Replacing the unobservable Pareto variables $V_{i,j}$ by the rank-based versions $\widehat V_{i,j} = n / (n + 1 - R_{ij})$ in Section~\ref{sec:etdf} yields an approximate sample
\[
  \widehat T_{i, \alpha} = \min_{j\in\alpha} \widehat V_{i,j},
  \qquad i = 1, \ldots, n,
\]
from the distribution of $T_\alpha$. Let $\widehat T_{(1), \alpha} \le \ldots \le \widehat T_{(n), \alpha}$ denote the order statistics of $\widehat T_{1,\alpha}, \ldots, \widehat{T}_{n,\alpha}$. The Hill estimator for $\eta_\alpha$ is defined as
\begin{equation}
  \label{eq:Hill}
  \widehat \eta_\alpha^H 
  = \frac{1}{k} \sum_{i=1}^k 
  \log \frac%
    {\widehat T_{(n-i+1), \alpha}}%
    {\widehat T_{(n-k), \alpha}}.
\end{equation}
Under the second-order regular variation conditions stated below, the asymptotic normality of $\widehat \eta_\alpha^H $ follows from \cite[proof of Theorem~2.1]{draisma2004bivariate}. The results in the cited reference cover the bivariate case only. In this section, we verify that they remain valid in any dimension $d\ge 2$, and we provide the general expression for  the asymptotic variance. Put $E_\alpha = [0, \infty]^\alpha \setminus \{ \bm{\infty}_\alpha \}$.


\begin{condition}
\label{as:draisma04}
For each $\alpha\subset\{1,\ldots,d\}$ with $|\alpha|\ge 2$, there exist functions $c_\alpha, c_{1,\alpha}: E_\alpha\to [0, \infty)$   such that $c_{1,\alpha}$ is neither constant nor a multiple of $c_\alpha$, and there exists $q_{1,\alpha}: (0, \infty) \to (0, \infty)$, with $q_{1,\alpha}(t) \to 0$ as $t\to 0$, such that, for all $\mb x \in E_\alpha$, we have
\begin{equation*}
  \lim_{t \to 0} 
    \left\{ 
      \frac%
	{\PP[ \forall j \in \alpha : 1 - F_j(X_j) \le tx_j]}%
	{\PP[ \forall j \in \alpha : 1 - F_j(X_j) \le t ]} 
      - c_\alpha(\mb x) 
    \right\} 
    \Big/ 
    q_{1, \alpha}(t) 
    = c_{1,\alpha}(\mb x).
\end{equation*}
\end{condition}

Under Condition~\ref{as:draisma04}, the function $q_\alpha(t) = \PP[ \forall j \in \alpha : 1 - F_j(X_j) \le t ]$ is regularly varying at $0$ with some index $1/\eta_\alpha$. Condition~\ref{as:draisma04} implies that the first-order condition~\eqref{eq:taildep-multi} holds with the same index $1/\eta_\alpha$.  In addition, $c_\alpha(\mb 1_\alpha) = 1$ and $c_\alpha$ is homogeneous of order $1/\eta_\alpha$, i.e., $c_\alpha(t \mb x ) =
t^{1/\eta_\alpha}c_\alpha(\mb x)$ for $t>0$, see \cite{draisma2001tail, draisma2004bivariate}. Under the regular variation assumption~\eqref{eq:reg-var}, we have  $\rho_\alpha = \lim_{t \to 0} q_\alpha(t)/t$, so that, under Condition~\ref{as:draisma04}, $\rho_\alpha > 0$ implies $\eta_\alpha = 1$, as in \cite{draisma2004bivariate} for the bivariate case. Finally, if $\rho_\alpha>0$, then $c_\alpha(\mb x )= r_\alpha(\mb x) /r_\alpha(\mb 1_\alpha) = r_\alpha(\mb x) / \rho_\alpha$. Note that in~\cite{draisma2004bivariate}, our $\rho_\alpha$ is denoted by $l$ for $\alpha = \{1,2\}$. 

The asymptotic variance of the Hill estimator \eqref{eq:Hill} involves a Gaussian process whose distribution depends on whether $\rho_\alpha=0$ or $\rho_\alpha > 0$. As in \cite{draisma2004bivariate}, introduce a centered Gaussian process $W_1$ on $E_\alpha$ with covariance function $\PE[ W_1(\mb x) \, W_1(\mb y) ]  = c_\alpha( \mb x\wedge \mb y)$ for $\mb x,\mb y \in E_\alpha$. Recall the stochastic process $Z_\alpha$ in \eqref{eq:Z_alpha(x)} and the random variable $G_\alpha = Z_\alpha(\mb{1}_\alpha)$ in \eqref{eq:G_alpha}. 

\begin{proposition}
\label{prop:normality-hill}
Let ${\mb X}_i = (X_{i,1}, \ldots, X_{i,d})$, for $i \in \{1,\ldots,n\}$, be an independent random sample from $F$, having continuous margins and satisfying \eqref{eq:reg-var}. Let $k = k(n)\to \infty$ as $n\to\infty$, while $k(n) = \oh(n)$. If Conditions~\ref{as:bias},~\ref{as:small-k},~\ref{as:partialDeriv}, and~\ref{as:draisma04} hold, then, as $n \to \infty$,
  \[
    \sqrt{k} \left( \widehat \eta_\alpha^H - \eta_\alpha \right)
    \wto \mathcal N(0,\sigma_{\alpha,H}^2), 
  \]
with $\sigma_{\alpha,H}^2 = \eta_\alpha^2 \Var\{\tilde W(\mb 1_\alpha) \}$, where $\tilde W(\mb x) = W_1(\mb x)$ if $\rho_\alpha = 0$ and $\tilde W(\mb{x}) = \rho_\alpha^{-1/2} Z_\alpha(\mb x)$ if $\rho_\alpha > 0$. In particular, if $\rho_\alpha > 0$, we have
\begin{equation}
\label{eq:sigma_alpha,H}
  \sigma_{\alpha,H}^2
  =
  \rho_\alpha^{-1} \Var(G_\alpha)
    = 1 - 2 \rho_\alpha 
    +
    \rho_\alpha^{-1}
    \sum_{j \in \alpha} \sum_{j' \in \alpha}
    \dot{\rho}_{j,\alpha}\dot{\rho}_{j',\alpha}
    \rho_{\{j,j'\}}.
\end{equation}
\end{proposition}

The proof of Proposition~\ref{prop:normality-hill} is based on the arguments developed in the proofs of \cite[Theorem~2.1]{draisma2004bivariate}, \cite[Theorem~3.2]{drees1998smooth}, and \cite[Example~3.1]{drees1998evindex}, which we gather in Appendix~\ref{sec:appendix}.   

%

Again, the unknown terms in \eqref{eq:sigma_alpha,H} may be replaced by their empirical counterparts, leading to an asymptotically consistent test. Recall $\dot{\rho}_{j,\alpha,n}$ in \eqref{eq:dotrho_j,alpha,n} and define
  \[
    \widehat\sigma^2_{\alpha, H} 
    = 1 - 2 \widehat\rho_\alpha 
    +
    \widehat\rho_\alpha^{-1}
    \sum_{j\in\alpha} \sum_{j'\in\alpha}
    \dot{\rho}_{j,\alpha,n}\dot{\rho}_{j',\alpha, n}
    \widehat\rho_{\{j,j'\}}.
  \]
  
The proof of the consistency of the variance estimator follows the same lines as the proofs of Propositions~\ref{prop:estim-sigma-kappa} and~\ref{prop:estimVarTerms} and is omitted.

\begin{cor}
\label{cor:testHill-estimVar}
Under the conditions of Proposition~\ref{prop:normality-hill}, if $\rho_\alpha > 0$, we have $\widehat\sigma_{\alpha,H}^2 = \sigma_{\alpha,H}^2 + \oh_{\PP}(1)$ as $n \to \infty$ and thus $\sqrt{k} (\widehat \eta_\alpha^H - 1) / \widehat{\sigma}_{\alpha, P} \wto \mathcal{N}(0,1)$, provided $\sigma_{\alpha,H}^2 > 0$.
\end{cor}

We may exploit Corollary~\ref{cor:testHill-estimVar} to test $H_0 : \rho_\alpha > 0$ in the same way as we did by using Peng's estimator in \eqref{eq:test-eta}: at significance level $\delta \in (0, 1)$, the null hypothesis is rejected in favour of $H_1: \eta_\alpha < 1$ when $\widehat\eta_\alpha^H < 1 - q_{1-\delta} k^{-1/2} \widehat{\sigma}_{\alpha,H}$.

\begin{remark}
The condition $\sigma_{\alpha,H}^2 > 0$ in Corollary~\ref{cor:testHill-estimVar} is satisfied whenever $0<\rho_\alpha<1$. Indeed, in \eqref{eq:sigma_alpha,H}, we have $\rho_{\{j,j'\}} \ge \rho_\alpha$ and $\dot{\rho}_{j,\alpha} \dot{\rho}_{j',\alpha} \ge 0$, whence $\sigma_{\alpha,H}^2$ $\ge 1 - 2\rho_\alpha + \sum_{(j,j')\in\alpha^2} \dot{\rho}_{j,\alpha} \dot{\rho}_{j',\alpha}$ $= 1 - 2 \rho_\alpha + \rho_\alpha^2 = (1 - \rho_\alpha)^2$. 
\end{remark}

\section{Simulation study}
\label{sec:simu-study}

Our aim is to compare the finite sample performance of the various tests proposed in Sections~\ref{sec:test-kappa}, \ref{sec:mult-extens-peng} and~\ref{sec:hill} within the framework of the CLEF algorithm, the pseudo-code of which is given in Appendix~\ref{sec:appendix-CLEF}. Three variants of the CLEF algorithm are obtained by varying the criterion according to which a subset $\alpha$ is declared as tail-dependent:
  $\widehat \kappa_\alpha > \kappa_{\min} - q_{\delta} \widehat \sigma_{\kappa,\alpha} / \sqrt{k}$ for CLEF-asymptotic;
 $\widehat\eta_{\alpha,P} > 1 - q_{\delta} \widehat \sigma_{\alpha,P} / \sqrt{k}$ for CLEF-Peng; and
 $\widehat\eta_{\alpha,H} > 1 - q_{\delta} \widehat \sigma_{\alpha,H} / \sqrt{k}$ for CLEF-Hill.
The original CLEF criterion was $\widehat{\kappa}_\alpha > C$ for some constant $C$ chosen by the user.
For completeness, the output of the DAMEX algorithm \cite{goixsparse} is included in the comparison. 

 

In practice, the dependence tests based on the tail dependence coefficient should not be carried out to the letter when the test statistic is not defined or when its estimated variance is infinite. Thus, in our experiments, CLEF-Peng and CLEF-Hill are modified so as to take into account additional, common-sense stopping criteria. A subset $\alpha$ will \emph{not} be part of the list returned by the algorithms under the following conditions: 



 \begin{enumerate}
 \item Concerning  CLEF-Hill, when $\widehat \rho_\alpha = 0$, that is, no extreme record  impacts  all coordinates in $\alpha$, the estimated variance of the Hill estimator of $\eta_\alpha$ is infinite. Therefore, $\widehat \rho_\alpha = 0$ is considered as a stopping criterion in  CLEF-Hill.
 \item  Concerning CLEF-Peng, when $\widehat r_{\alpha}(\mb 2_\alpha) = \widehat r_{\alpha}(\mb 1_\alpha)$, the Peng estimator \eqref{eq:multi-peng} is ill-defined. Such a case arises when there are very few points in the joint tail within the subspace generated by $\alpha$.
When the estimated derivatives $\dot{\rho}_{j,\alpha,n}$  are close to zero, and when $ \widehat\rho_\alpha\ll 1$,  the estimated variance $\widehat \sigma_{\alpha,P}^2$ in~\eqref{eq:Hatvar-hateta} becomes large, preventing rejection of the null hypothesis. 
To prevent these issues, each of the conditions $\widehat\rho_\alpha< 0.05$  and $\widehat r_{\alpha}(\mb 2_\alpha) = \widehat r_{\alpha}(\mb 1_\alpha)$ are declared as a stopping criterion in CLEF-Peng.
\end{enumerate}



\paragraph{Experimental setting.}
CLEF~\cite{chiapinofeature} is designed to face situations where  DAMEX~\citep{goixsparse} fails to exhibit a clear-cut dependence  structure. A major issue reported in~\cite{chiapinofeature} for certain hydrological data is the high variability of the groups of features for which large values occur simultaneously. Because of this, the empirical exponent measure $\widehat\mu$ assigns low mass to any sub-region partitioning the sample space, see Remark~\ref{rm:relationship-goix}. The empirical finding motivating the latter work is that the various subsets $\alpha$ involved in simultaneous extreme records could nevertheless be clustered, meaning that many of them have a significant intersection, whereas many symmetric differences comprise just a single or at most a few features.

A natural assumption in this context is that a `true' list of dependent subsets $\mathcal{M} = \{\alpha_1,\ldots, \alpha_K\}$ exists such that $\mu(\mathcal{C}_\alpha)>0$ for $\alpha\in\mathcal{M}$ and that noisy features are involved in each extreme event. Observed large records then concern groups of the kind $\alpha' = \alpha\cup\{j\}$, where $\alpha\in\mathcal{M}$ and $j\in\dd\setminus\alpha$.
  
In our experiments, datasets are generated as follows: The  dimension is fixed to $d=100$. A family of `true' dependent subsets $\mathcal{M}= \{\alpha_1,\ldots,\alpha_K\}$ of cardinality $K = 80$ is randomly chosen: the subset sizes $|\alpha|$ follow a truncated geometric distribution, with a maximum subset size set to $8$. For simplicity, we forbid nested subsets, so $\alpha_j\not\subset\alpha_k$ whenever $j \ne k$. The maximal elements of $\mathbb{M} = \{\alpha\subset\dd \mid \rho_\alpha>0\}$ are then precisely the elements of $\mathcal{M}$, as explained in  Remark~\ref{rm:relationship-goix}. Finally, two different subsets may have at most two features in common. 

Once the dependence structure $\mathcal{M}$ has been fixed, the data $\mb{X}_1, \ldots, \mb{X}_n$ are sampled independently from  $d$-dimensional asymmetric logistic distributions \cite{tawn1990modelling}, using Algorithm $2.2$ in
\cite{stephenson2003simulating}. The underlying `true' distribution function is 
\begin{equation}
\label{eq:asym-logistic}
  G(\mb x) = \exp\biggl[- \sum_{m=1}^K\Bigl\{ \sum_{j\in\alpha_{m}} (\lvert{\cal A}(j)\rvert x_j)^{-1/w_{\alpha_m}} \Bigr\}^{w_{\alpha_m}}\biggr],
\end{equation}
where ${\cal A}(j) = \{\alpha \in \mathcal{M} \mid j \in \alpha\}$ and $w_{\alpha_m}$ is a dependence parameter which is set to $0.1$ in our simulations. Actually, to mimic the noisy situation described above, each point $\mb X_i$ is simulated according to a slightly different version, $G_i$, of $G$. For each $i = 1, \ldots, n$ and $k = 1, \ldots, K$, we randomly select an additional `noisy feature' $j_{i,k} \in \{1, \ldots, d\} \setminus \alpha_k$ and set $\alpha_{i,k}' = \alpha_k \cup \{ j_{i,k} \}$. Then $\mathcal{M}_i' = \{ \alpha_{i,1}', \ldots, \alpha_{i,K}' \}$ is the collection of `noisy subsets' for $\mb{X}_i$ and $G_i(\mb{x})$ is as in \eqref{eq:asym-logistic} with $\mathcal{A}(j)$ replaced by $\mathcal{A}'_i(j) = \{\alpha' \in \mathcal{M}'_i \mid j \in \alpha'\}$.

\paragraph{Results.}
We generate datasets of size $n=5\mathrm{e}4$ and $n=1\mathrm{e}5$. For each sample size, $50$ independent datasets are simulated according to the
procedure summarized in the preceding paragraph. We compare the average performance of the three proposed versions of CLEF, together with the original CLEF and DAMEX algorithms,  for different choices of $k$
and confidence level $\delta$.  



\begin{table}[!ht]
    \caption{ Average number of recovered clusters and errors of
      CLEF-asymptotic
      ($\kappa_{\min}=0.08$), CLEF-Peng, CLEF-Hill, CLEF and DAMEX on
      $50$ datasets. Confidence level for the tests: $\delta=0.001$. Standard deviations over the 50 samples in brackets. Bold face indicates the best performing
      algorithm on average for a given $n$ and a given choice of $k/n$, the proportion of extreme data used.}
      \label{hpc-noise-table-deltaBig}
  \centering
      \begin{tabular}{@{}l@{\quad}ccccc@{}}
      \toprule
      $n=5\mathrm{e}4$ & $k/n$ & recovered &  subset  errors & superset  errors & other errors \\
      \midrule
        CLEF-asymptotic 
            & 0.003 & 71.1 (3.0) & 7.4 (4.7) & 5.1 (2.1) & 28.0 (13.3)\\
              & 0.005 & 73.0 (3.7) & 8.0 (6.3) & 2.4 (1.7) & 14.6 \phantom{0}(8.9) \\[.7ex]
      CLEF-Peng 
                & 0.003 & {\bf 79.70 (0.7)} & {\bf 1.00 (2.5)} & {\bf 0. (0.)} & 3.9 (2.7)\\
                & 0.005 & {\bf 79.98 (0.1)} & {\bf 0.06 (0.4)} & {\bf 0. (0.)} & 0.9 (0.9)\\[.7ex]
      CLEF-Hill  
                & 0.003 & 79.0 (1.4) & 2.4 (3.5) & 0.04 (0.2) & 17.9 (7.0)\\
                & 0.005 & 75.7 (2.4) & 9.2 (6.8) & {\bf 0. (0.)} & {\bf 0. (0.)}\\[.7ex]
           CLEF 
           & 0.003 & 69.9 (4.4) & 16.2 (8.1) & 0.5 (0.6) & {\bf 2.3 (2.2)}\\
           & 0.005 & 75.0 (3.6) & \phantom{0}8.1 (6.4) & 0.2 (0.5) & 0.9 (1.2) \\[.7ex]
      DAMEX 
            & 0.003 & 0.6 (0.2) & 1.7 (1.4) & 32.9 (5.6) & 45.4 (5.9) \\
            & 0.005 & 0.1 (0.4) & 2.4 (1.5) & 18.3 (5.5) & 59.1 (5.9) \\
	\midrule
$n=1\mathrm{e}5$ &&&& &\\                
	\midrule
        CLEF-asymptotic  & 0.003 & 73.2 (3.7) & \phantom{0}9.5 (6.7) & 0.9 (0.8) & 4.7 (2.7)\\
                & 0.005 & 72.6 (4.4) & 11.7 (7.6) & 0.1 (0.4) & 0.5 (0.9) \\[.7ex]
        CLEF-Peng & 0.003 & {\bf 79.9 (0.2)} & {\bf 0.2 (1.0)} & {\bf 0. (0.)} & 0.1 (0.4)\\
                & 0.005 & {\bf 80.0 (0.)\phantom{0}} & {\bf 0.\phantom{0} (0.)\phantom{0}} & {\bf 0. (0.)} & {\bf 0. (0.)}\\[.7ex]
        CLEF-Hill & 0.003 & 77.0 (2.0) & \phantom{0}6.1 \phantom{0}(4.6) & {\bf 0. (0.)} & {\bf 0. (0.)}\\
                & 0.005 & 67.2 (4.8) & 22.8 (10.4) & {\bf 0. (0.)} & {\bf 0. (0.)} \\[.7ex]
        CLEF & 0.003 & 75.2 (3.2) & 7.5 (5.9) & 0.0\phantom{0} (0.2) & 0.2\phantom{0} (0.5)\\
              & 0.005 & 77.9 (2.3) & 3.2 (3.9) & 0.02 (0.1) & 0.02 (0.1) \\[.7ex]
        DAMEX & 0.003 & 0.04 (0.2) & 1.3 (1.0) & 24.4 (6.7) & 54.2 (7.0) \\
              & 0.005 & 0.1\phantom{0} (0.3) & 1.9 (1.6) & 10.3 (3.7) & 67.6 (4.7) \\
      \bottomrule
    \end{tabular} 
\end{table}

\begin{table}[!ht]
   \caption{Same setting as Table \ref{hpc-noise-table-deltaBig} with $\delta=0.0001$}
   \label{hpc-noise-table-deltaSmall}
  \centering
  \begin{tabular}{@{}l@{\quad}ccccc}
    \toprule
    $n=5\mathrm{e}4$ & $k/n$ & recovered &  subset  errors & superset  errors & other errors \\
    \midrule
    CLEF-asymptotic & 0.003 & 71.8 (2.4) & 2.3 (2.5) & 7.8 (2.8) & 41.9 (19.3) \\
            & 0.005 & 73.5 (2.8) & 3.7 (3.8) & 4.8 (2.5) & 25.8 (12.2) \\[.7ex]
    CLEF-Peng 
            & 0.003 & {\bf 79.7 (0.7)} & {\bf 1.0 (2.5)} & {\bf 0. (0.)} & 3.9 (2.7)\\
            & 0.005 & {\bf 80.0 (0.1)} & {\bf 0.1 (0.4)} & {\bf 0. (0.)} & 0.9 (0.9)\\[.7ex]
    CLEF-Hill 
            & 0.003 & 79.5 (0.8) & {\bf 0.3 (1.1)} & 0.5 (0.8) & 142.2 (33.2) \\
            & 0.005 & 79.2 (1.0) & 1.6 (2.3) & {\bf 0. (0.)} & {\bf 0.2 (0.5)} \\[.7ex]
           CLEF 
           & 0.003 & 69.9 (4.4) & 16.2 (8.1) & 0.5 (0.6) & 2.3 (2.2)\\
           & 0.005 & 75.0 (3.6) & \phantom{0}8.1 (6.4) & 0.2 (0.5) & 0.9 (1.2) \\[.7ex]
      DAMEX 
            & 0.003 & 0.6 (0.2) & 1.7 (1.4) & 32.9 (5.6) & 45.4 (5.9) \\
            & 0.005 & 0.1 (0.4) & 2.4 (1.5) & 18.3 (5.5) & 59.1 (5.9) \\
      \midrule
    $n=1\mathrm{e}5$ & &&&& \\
    \midrule
    CLEF-asymptotic  & 0.003 & 75.7 (2.8) & 3.7 (3.8) & 2.0 (1.4) & 11.0 (5.5) \\
            & 0.005 & 76.0 (2.9) & 5.6 (4.5) & 0.4 (0.7) & \phantom{0}1.9 (1.9) \\[.7ex]
    CLEF-Peng & 0.003 & {\bf 79.9 (0.2)} & {\bf 0.2 (1.0)} & {\bf 0. (0.)} & 0.1 (0.4)\\
            & 0.005 & {\bf 80.\phantom{0} (0.)\phantom{0}} & {\bf 0.\phantom{0} (0.)\phantom{0}} & {\bf 0. (0.)} & {\bf 0. (0.)}\\[.7ex]
    CLEF-Hill & 0.003 & 79.5 (1.0) & 1.2 (2.3) & {\bf 0. (0.)} & {\bf 0.1 (0.2)} \\
            & 0.005 & 75.4 (2.8) & 8.7 (5.2) & {\bf 0. (0.)} & {\bf 0.\phantom{0} (0.)\phantom{0}} \\[.7ex]
          CLEF & 0.003 & 75.2 (3.2) & 7.5 (5.9) & 0.0\phantom{0} (0.2) & 0.2\phantom{0} (0.5)\\
           & 0.005 & 77.9 (2.3) & 3.2 (3.9) & 0.02 (0.1) & 0.02 (0.1) \\[.7ex]
      DAMEX & 0.003 & 0.04 (0.2) & 1.3 (1.0) & 24.4 (6.7) & 54.2 (7.0) \\
            & 0.005 & 0.1\phantom{0} (0.3) & 1.9 (1.6) & 10.3 (3.7) & 67.6 (4.7) \\
    \bottomrule
    \end{tabular}
\end{table}

Tables~\ref{hpc-noise-table-deltaBig}
and~\ref{hpc-noise-table-deltaSmall} gather the results 
for a  confidence level $\delta$ equal to $0.001$ and $0.0001$, respectively.
In both tables, 
the results obtained with the original version of CLEF
and DAMEX are included in the comparison with an identical choice of tuning parameters, so that the last two lines of the two tables are the same. 
In CLEF,  the threshold $C$ was chosen  by trial and error  in the interval $(0, \kappa_{\min})$, namely $C=0.05$. 
Imposing that $C < \kappa_{\min}$ is intended to reproduce the effect of the variance term upon the stopping criterion  in CLEF-asymptotic.   
In DAMEX, the $80$ subsets with highest empirical mass are retained and  the subspace thickening parameter $\epsilon$ is set to the default  value of $0.1$, following the guidelines of the authors. 

Each algorithm produces a list, $\widehat{\mathbb{M}}$, of groups of features $\alpha \in \{1, \ldots, d\}$. This list is to be compared with the one of $K = 80$ `true' subsets $\mathcal{M}$. The performance of each algorithm is measured in terms of two criteria: the number of `true' subsets $\alpha\in\mathcal{M}$ that appear in $\widehat{\mathbb{M} }$ (third column of Tables~\ref{hpc-noise-table-deltaBig} and~\ref{hpc-noise-table-deltaSmall}); the number of `errors', that is, the subsets $\alpha\in \widehat{\mathbb{M}}$ that do not belong to $\mathcal{M}$. These can be understood as `false positives'. Among these errors, we make the distinction between those which are  respectively proper subsets (fourth column of Tables~\ref{hpc-noise-table-deltaBig} and~\ref{hpc-noise-table-deltaSmall}) or proper supersets (fifth column) of some true $\beta\in\mathcal{M}$, and the other errors (sixth column).

CLEF-Peng obtains  the best overall scores for both values of $\delta$, but as explained above, a special treatment is reserved for the case $\widehat\rho_\alpha \le 0.05$, and this threshold constitutes an arbitrary tuning parameter, which can impact the performance significantly. On the other hand, CLEF-Hill does not require any other adjustment than for the special case $\widehat\rho_\alpha=0$ and performs nearly as well as CLEF-Peng with $\delta=0.0001$ and $k/n=0.005$. In addition, CLEF-Hill outperforms all the other methods. In particular, CLEF-asymptotic is globally less accurate than CLEF-Peng and CLEF-Hill. This  reflects the fact that the  null hypothesis in this algorithm  involves an arbitrary $\kappa_{\min}>0$ fixed by the user. Our own choice $\kappa_{\min} = 0.08$ was fixed by trial and error, which is straightforward with synthetic data and could also be achieved by cross-validation in a real use case.  Finally, as expected, DAMEX obtains very low scores, because it is not designed to handle the addition of noisy features, as explained earlier.
\section{Conclusion}
\label{sec:conclusion}

In this work, we propose three variants of the CLEF algorithm \citep{chiapinofeature}, replacing the heuristic criterion in the original version with a formal test for asymptotic dependence, and this for all possible subsets of features among $\{1,\ldots,d\}$. As in the original CLEF implementation, only a small proportion of all $2^d-1$ subsets has to be examined, while the computational complexity for each such subset is low. Experimental results indicate that the CLEF algorithm is most effective when based on a test constructed from an extension of the Hill estimator \citep{draisma2004bivariate} of the multivariate coefficient of tail dependence.

The procedure we propose is nonparametric and rank-based. Parametric approaches, based for instance on the nested asymmetric logistic distribution \citep{tawn1990modelling}, could have a greater sensitivity, at the cost of increased model risk and greater computational complexity. We have also assumed that the observations are serially independent; in the contrary case, the asymptotic variances of the various estimator need to be estimated by some form of bootstrap, which, in high dimensions, poses important theoretical and computational challenges; see \citep{bucher+d:2013} for the bivariate and serially independent case.

\appendix

\section{Proofs}
\label{sec:appendix}

\begin{proof}[Proof of Proposition~\ref{prop:rnx}]
For $\emptyset \ne \alpha \subset \dd$ and $\bm{x} \in [0, \infty)^\alpha$, put 
\begin{align*} 
  L_\alpha(\bm{x}) 
  &= 
  \{ \bm{y} \in [0, \infty]^d \mid \exists j \in \alpha : y_j < x_j \}, \\
  R_\alpha(\bm{x}) 
  &= 
  \{ \bm{y} \in [0, \infty]^d \mid \forall j \in \alpha : y_j < x_j \}.
\end{align*}
If $\alpha = \dd$, then just write $L$ rather than $L_{\dd}$. Note that $L_\alpha(\bm{x}) = L(\bm{x} \bm{e}_\alpha)$ with $\bm{e}_\alpha = (\un_\alpha(j))_{j=1}^d$ and that $L_{\{j\}}(x_j) = R_{\{j\}}(x_j)$ and thus $W(L_{\{j\}}(x_j)) = W_{\{j\}}(x_j)$. \citet[Theorem~4.6]{einmahl2012m} show that, in the space $\ell^\infty([0, T]^d)$ and under Conditions~\ref{as:bias},~\ref{as:small-k} and~\ref{as:partialDeriv}, we have weak convergence 
\[ 
  \sqrt{k}\{ \widehat{\ell}(\bm{x}) - \ell(\bm{x}) \} 
  \wto 
  W(L(\bm{x})) - \sum_{j=1}^d \partial \ell_j(\bm{x}) W_{\{j\}}(x_j)
\]
as $n \to \infty$. Here, we have taken a version of the Gaussian process $W$ such that the trajectories $\mb{x} \mapsto W(L(\bm{x}))$ are continuous almost surely.

As in \eqref{eq:ell_beta2r_alpha}, we have, for $\emptyset \ne \alpha \subset \dd$ and $\bm{x} \in [0, \infty)^\alpha$, the identity
\[
  \widehat{r}_\alpha(\bm{x}) 
  = \sum_{\emptyset \ne \beta \subset \alpha} (-1)^{|\beta|+1} \widehat\ell( \bm{x}_\beta \bm{e}_\beta )
\]
where $\bm{x}_\beta = (x_j)_{j \in \beta}$. Hence, we can view the vector $(\sqrt{k}( \widehat{r}_\alpha - r_\alpha ))_{\emptyset \ne \alpha \subset \dd}$ as the result of the application to $\sqrt{k}(\widehat{\ell}-\ell)$ of a bounded linear map from the space $\ell^\infty([0, T]^d)$ to the product space $\prod_{\emptyset \ne \alpha \in \dd} \ell^\infty([0, T]^\alpha)$. By the continuous mapping theorem, we obtain, in the latter space, the weak convergence
\[
  \sqrt k\left\{\widehat r_\alpha(\mb x) - r_\alpha(\mb x)\right\}
  \wto
  \sum_{\emptyset \ne \beta \subset \alpha} 
  (-1)^{|\beta|+1} 
  \left\{
  W(L_\beta(\bm{x}_\beta))
  -
  \textstyle\sum_{j=1}^d \partial_j \ell_\beta(\bm{x}_\beta) W_{\{j\}}(x_j \un_\beta(j))
  \right\}.
\]
Here we used $\ell(\bm{x}_\beta \bm{e}_\beta) = \ell_\beta(\bm{x}_\beta)$.

The set-indexed process $W$ satisfies the remarkable property that $W(A \cup B) = W(A) + W(B)$ almost surely whenever $A$ and $B$ are disjoint Borel sets of $[0, \infty]^d \setminus \{ \bm{\infty} \}$ that are bounded away from $\bm{\infty}$: indeed, \eqref{eq:covW} implies $\PE[\{W(A \cup B) - W(A) - W(B)\}^2] = 0$. It follows that the trajectories of $W$ obey the inclusion-exclusion formula, so that, for $\emptyset \ne \alpha \subset \dd$ and $\bm{x} \in [0, \infty)^\alpha$, we have, almost surely,
\begin{align*}
  \sum_{\emptyset \ne \beta \subset \alpha} 
  (-1)^{|\beta|+1} 
  W(L_\beta(\bm{x}_\beta))
  &=
  \sum_{\emptyset \ne \beta \subset \alpha} 
  (-1)^{|\beta|+1} 
  W\left( \textstyle\bigcup_{j\in\beta} R_{\{j\}}(x_j) \right) \\
  &=
  W\left( \textstyle\bigcap_{j\in\alpha} R_{\{j\}}(x_j) \right)
  =
  W(R_\alpha(\bm{x}))
  =
  W_\alpha(\bm{x}).
\end{align*}
We can make this hold true almost surely jointly for all such $\alpha$ and $\bm{x}$: first, consider points $\bm{x}$ with rational coordinates only and then consider a version of $W$ by extending $W_\alpha$ to points $\bm{x}$ with general coordinates via continuity. Similarly, since $W_{\{j\}}(0) = W(\emptyset) = 0$ almost surely, we have
\begin{align*}
  \sum_{\emptyset \ne \beta \subset \alpha} (-1)^{|\beta|+1} 
  \sum_{j=1}^d \partial_j \ell_\beta(\bm{x}_\beta) W_{\{j\}}(x_j \un_\beta(j))
  &=
  \sum_{j \in \alpha} \sum_{\beta : j \in \beta \subset \alpha} \partial_j \ell_\beta(\bm{x}_\beta) W_{\{j\}}(x_j) \\
  &=
  \sum_{j \in \alpha} \partial r_j(\bm{x}) W_{\{j\}}(x_j).
\end{align*}
We have thus shown weak convergence as stated in \eqref{eq:Z_alpha(x)}.
\end{proof}

\begin{proof}[Proof of Corollary~\ref{prop:asymptotic-hatRho}]
  The weak convergence statement \eqref{eq:G_alpha} is a special case of \eqref{eq:Z_alpha(x)}: set $\bm{x} = \bm{1}_\alpha$. The covariance formula \eqref{eq:covG} follows from the fact that
  
\begin{align*}
  \PE[W_\alpha(\bm{1}_\alpha) W_{\alpha'}(\bm{1}_{\alpha'})]
  &=
  \Lambda( \{ \bm{y} \in [0, \infty]^d \mid \forall i \in \alpha \cup \alpha' : y_i < 1 \} ) \\
  &=
  \mu( \{ \bm{u} \in [0, \infty)^d \mid \forall i \in \alpha \cup \alpha' : u_i > 1 \} )
  =
  \rho_{\alpha\cup\alpha'};
\end{align*}
the first equality follows from \eqref{eq:covW} and the last one from \eqref{eq:rho_alpha}. We obtain \eqref{eq:covG} by expanding $G_\alpha = Z_\alpha(\bm{1}_\alpha)$ using \eqref{eq:Z_alpha(x)} and working out $\PE[ G_\alpha G_{\alpha'} ]$ with the above identity.
\end{proof}

\begin{proof}[Proof of Proposition~\ref{theo:asymptot-kappa}] 
Let $\alpha = \{\alpha_1, \ldots, \alpha_S \} \subset \{1, \ldots, d\}$ with $S = \lvert \alpha \rvert \ge 2$ and such that $\mu(\Delta_\alpha) > 0$. In view of~\eqref{eq:rewriteKappa}, we have $\kappa_\alpha = g_\alpha( \theta_\alpha )$ and $\widehat{\kappa}_\alpha = g_\alpha( \widehat{\theta}_\alpha )$ where $\theta_\alpha = ( \rho_\alpha, \rho_{\alpha\setminus\alpha_1}, \ldots, \rho_{\alpha\setminus\alpha_S} )$, $\widehat\theta_\alpha = (\widehat\rho_\alpha, \widehat\rho_{\alpha\setminus\alpha_1}, \ldots, \widehat\rho_{\alpha\setminus\alpha_S})$, and
\begin{equation}
  \label{eq:transfoKappa}
  g_\alpha( x_0, x_1, \ldots, x_{S}) = \frac{x_0}{\sum_{j=1}^{S} x_j - (S-1) x_0}, \qquad x\in[0,\infty)^{1+S}.  
\end{equation}
Let $\nabla g_\alpha(x)$ denote the gradient vector of $g_\alpha$ evaluated $x$ and let $\langle \, \cdot \, , \, \cdot \, \rangle$ denote the scalar product in Euclidean space. Proposition~\ref{prop:asymptotic-hatRho} combined with the delta method as in \cite[Theorem~3.1]{van2000asymptotic} gives, as $n \to \infty$,
\begin{align*}
  \sqrt{k} ( \widehat\kappa_\alpha - \kappa_\alpha )
  =
  \sqrt{k} \{ g_\alpha( \widehat{\theta}_\alpha ) - g_\alpha( \theta_\alpha ) \}
  &=
  \left\langle 
    \nabla g_\alpha(\theta_\alpha), \, 
    \sqrt{k} (\widehat{\theta}_\alpha - \theta_\alpha) 
  \right\rangle
  +
  \oh_{\PP}(1)
  \\
  &\wto
  \left\langle 
    \nabla g_\alpha(\theta_\alpha), \, 
    (G_\alpha,G_{\alpha\setminus\alpha_1},\ldots,G_{\alpha\setminus\alpha_S})
  \right\rangle,
\end{align*}
the weak convergence holding jointly in $\alpha$ by Slutsky's lemma and Proposition~\ref{prop:asymptotic-hatRho}. The partial derivatives of $g_\alpha$ are
\begin{align*}
  \frac{\partial g}{\partial x_0}( x ) 
  &= \frac{\sum_{j=1}^S  x_j }{\{ \sum_{j=1}^S  x_j - (S-1) x_0 \}^2}, \\
  \frac{\partial g}{\partial x_j}( x )
  &= \frac{- x_0}{\{ \sum_{j=1}^S  x_j - (S-1) x_0 \}^2},
    \qquad j = 1, \ldots, S.
\end{align*}
Evaluating these at $x = \theta_\alpha$ and using $\sum_{j \in \alpha} \rho_{\alpha \setminus j} - (S-1) \rho_\alpha = \mu(\Delta_\alpha)$ as in~\eqref{eq:B-Gamma} and~\eqref{eq:rewriteKappa}, we find that
\[
  \left\langle 
    \nabla g_\alpha(\theta_\alpha), \, 
    (G_\alpha,G_{\alpha\setminus\alpha_1},\ldots,G_{\alpha\setminus\alpha_S})
  \right\rangle
  =
  \mu(\Delta_\alpha)^{-2}
  \left\{
    \left( \textstyle\sum_{j\in\alpha}\rho_{\alpha\setminus j} \right) G_\alpha
    - \rho_\alpha \textstyle\sum_{j\in\alpha} G_{\alpha\setminus j}
  \right\},
\]
in accordance to the right-hand side in \eqref{eq:kappalimit}.

To calculate the asymptotic variance $\sigma_{\kappa,\alpha}^2$, we introduce a few abbreviations: we write $R_\beta = R_\beta(\bm{1}_\beta)$ and $W_\beta^\cap = W_\beta(\bm{1}_\beta) = W(R_\beta)$ for $\emptyset \ne \beta \in \dd$ and we put $W_j = W_{\{j\}}(1)$ for $j = 1, \ldots, d$, so that $G_\alpha = W_\alpha^\cap - \sum_{j\in\alpha}\dot{\rho}_{j,\alpha} W_j$. We find
\begin{align*}
  H_\alpha
  &= 
  \Big( \sum_{i\in\alpha}\rho_{\alpha\setminus i} \Big) G_\alpha 
  - 
  \rho_\alpha \sum_{i\in\alpha} G_{\alpha\setminus i} \\
  &=  
  \Big(\sum_{i\in\alpha}\rho_{\alpha\setminus i}\Big)
  \Big( W^\cap_{\alpha} - \sum_{j\in\alpha}\dot{\rho}_{j,\alpha}W_{j}\Big) 
  - 
  \rho_\alpha \sum_{i\in\alpha}
    \Big( W^\cap_{\alpha\setminus i} - \sum_{j\in\alpha\setminus i} \dot{\rho}_{j,\alpha\setminus i} W_{j} \Big).
\end{align*}
From the proof of Proposition~\ref{prop:rnx}, recall that $W(A \cup B) = W(A) + W(B)$ almost surely for disjoint Borel sets $A$ and $B$ of $[0, \infty]^d \setminus \{\bm\infty\}$ bounded away from $\bm{\infty}$; moreover, for such $A$ and $B$, the variables $W(A)$ and $W(B)$ are uncorrelated. Since $R_{\alpha \setminus i}$ is the disjoint union of $R_\alpha$ and $R_{\alpha \setminus i} \setminus R_\alpha$, we have therefore $W_{\alpha\setminus i}^\cap = W_\alpha^\cap + W(R_{\alpha\setminus i} \setminus R_\alpha)$ almost surely. In addition, $\sum_{i \in \alpha} \rho_{\alpha\setminus i} = \mu(\Delta_\alpha) + (S-1)\rho_\alpha$ by \eqref{eq:B-Gamma} applied to $\nu = \mu$. As a consequence,
\[
  H_\alpha
  = 
  \{ \mu(\Delta_\alpha) - \rho_\alpha \}
  W_{\alpha}^\cap
  - \rho_\alpha \sum_{j \in\alpha} W(R_{\alpha\setminus j} \setminus R_\alpha)
  +  \sum_{j\in\alpha} K_{\alpha, j} W_j
\]
where
\[
  K_{\alpha,j}
  =
  \rho_\alpha\Big(\sum_{i\in\alpha\setminus j}\dot{\rho}_{j,\alpha\setminus i}\Big) 
  -
  \Big(\sum_{i\in\alpha}\rho_{\alpha\setminus i}\Big)\dot{\rho}_{j,\alpha},
  \qquad j \in \alpha.
\]
The $S+1$ variables $W_\alpha^\cap = W(R_\alpha)$ and $W(R_{\alpha\setminus j} \setminus R_\alpha)$, $j \in \alpha$, are all uncorrelated, since they involve evaluating $W$ at disjoint sets; $W_j = W(R_{\{j\}})$ is uncorrelated with $W(R_{\alpha \setminus j} \setminus R_\alpha)$, for the same reason. Moreover, $\PE[ W_\alpha^\cap W_j ] = \Lambda(R_\alpha \cap R_{\{j\}}) = \Lambda(R_\alpha) = \rho_\alpha$ and similarly $\PE[ W(R_{\alpha\setminus i} \setminus R_\alpha) W_j ] = \Lambda(R_{\alpha\setminus i} \setminus R_{\alpha}) = \rho_{\alpha \setminus i} - \rho_\alpha$ if $i,j\in\alpha$ and $i \ne j$. Hence

\begin{multline*}
  \Var(H_\alpha)
  = 
  \{\mu(\Delta_\alpha) - \rho_\alpha\}^2 \rho_\alpha  
  +
  \rho_\alpha^2 \sum_{j \in \alpha} (\rho_{\alpha \setminus j} - \rho_\alpha)
  + 
  \sum_{i,j \in \alpha} K_{\alpha,i} K_{\alpha,j} \rho_{\{i,j\}} \\
  + 
  \{\mu(\Delta_\alpha) - \rho_\alpha\} \rho_\alpha\sum_{j\in\alpha}K_{\alpha,j}
  - 
  \rho_\alpha \sum_{j\in\alpha}K_{\alpha,j} \sum_{i\in\alpha\setminus j} (\rho_{\alpha\setminus i} - \rho_\alpha).
\end{multline*}
As $\sum_{j \in \alpha} (\rho_{\alpha \setminus j} - \rho_\alpha) = \mu(\Delta_\alpha)- \rho_\alpha$ and $\sum_{i \in \alpha \setminus j}(\rho_{\alpha \setminus i} - \rho_\alpha) = \mu(\Delta_\alpha) - \rho_{\alpha,j}$, we get

\begin{multline}
\label{eq:varmuZ}
  \Var(H_\alpha)
  = \{\mu(\Delta_\alpha) - \rho_\alpha\}\rho_\alpha
    \Big\{ \mu(\Delta_\alpha) + \sum_{j\in\alpha}K_{\alpha,j} \Big\}
    + \sum_{i,j\in\alpha}K_{\alpha,i}K_{\alpha,j} \rho_{\{i,j\}} \\
    -\rho_{\alpha} \sum_{j \in\alpha}K_{\alpha,j}\{\mu(\Delta_\alpha) - \rho_{\alpha\setminus j} \}.
\end{multline}
Recall $\kappa_\alpha(\bm{x})$ in \eqref{eq:kappax}. We have
\begin{align*}
  \frac{\partial}{\partial x_j} 
  \left(
    \frac{1}{\kappa_\alpha(\bm{x})}
  \right)_{\bm{x} = \bm{1}_\alpha}
  &=
  \frac{\partial}{\partial x_j} 
  \left(
    \frac{\sum_{i \in \alpha} r_{\alpha \setminus i}(\bm{x}_{\alpha \setminus i})}{r_{\alpha}(\bm{x})}
  \right)
  \\
  &=
  \rho_\alpha^{-2}
  \bigg(
    \rho_\alpha \sum_{i \in \alpha \setminus j} \dot{\rho}_{j, \alpha \setminus i}
    -
    \dot{\rho}_{j, \alpha} \sum_{i \in \alpha} \rho_{\alpha \setminus i}
  \bigg)
  =
  \rho_\alpha^{-2} K_{\alpha, j}.
\end{align*}
It follows that $\dot{\kappa}_{j,\alpha} = - \rho_{\alpha}^{-2} K_{\alpha, j} / (1/\kappa_{\alpha})^2 = - K_{\alpha,j} / \mu(\Delta_\alpha)^2$.
By \eqref{eq:varmuZ}, we find that $\sigma_{\kappa,\alpha}^2 = \mu(\Delta_\alpha)^{-4} \Var( H_{\alpha} )$ is equal to the right-hand side of \eqref{eq:sigmakappa2}.
\end{proof}     

\begin{proof}[Proof of Proposition~\ref{prop:estim-sigma-kappa}]
We only need to prove that $\widehat{\sigma}^2_{\kappa,\alpha} = \sigma_{\kappa, \alpha}^2 + \oh_{\PP}(1)$ as $n \to \infty$. In view of the expressions~\eqref{eq:sigmakappa2} and~\eqref{eq:Hatvar-kappa} for $\sigma_{\kappa, \alpha}^2$ and $\widehat{\sigma}_{\kappa,\alpha}$, it is enough to show that $\dot{\kappa}_{j,\alpha, n} = \dot{\kappa}_{\alpha,j} + \oh_{\PP}(1)$, with $\dot{\kappa}_{j,\alpha, n}$ in \eqref{eq:dotkappan}; indeed, Corollary~\ref{prop:asymptotic-hatRho} already gives consistency of $\widehat{\mu}(\Delta_\alpha)$ and $\widehat{\rho}_\beta$. Now since $2^{-1} k^{1/4} \{ \kappa_{\alpha}(\mb
1_\alpha + k^{-1/4}\mb e_j) -  \kappa_{\alpha}(\mb 1_\alpha - k^{-1/4}\mb e_j ) \} \to \dot{\kappa}_{\alpha,j}$ as $n \to \infty$, a sufficient condition is that for some $\epsilon>0$, 
\begin{equation}
  \label{eq:toshowKappa4}
  \sup_{[1-\epsilon,2+\epsilon]^\alpha} k^{1/4}\big|\widehat \kappa_\alpha(\mb x) - \kappa_\alpha(\mb x)\big|
  = \oh_{\PP}(1), \qquad n \to \infty.
\end{equation}
In turn, \eqref{eq:toshowKappa4} follows from weak convergence of $k^{1/2}( \widehat{\kappa}_\alpha - \kappa_\alpha )$ as $n \to \infty$ in the space $\ell^{\infty}([1-\eps, 1+\eps]^\alpha)$. In light of the expressions of $\widehat{\kappa}_\alpha$ and $\kappa_\alpha$ in terms of the (empirical) joint tail dependence functions $\widehat{r}_\beta$ and $r_\beta$, respectively, weak convergence of $k^{1/2}( \widehat{\kappa}_\alpha - \kappa_\alpha )$ follows from Proposition~\ref{prop:rnx} and the functional delta method \citep[Theorem~20.8]{van2000asymptotic}. The calculations are similar to the ones for the Euclidean case in the proof of Proposition~\ref{theo:asymptot-kappa}; an extra point to be noted is that if $\alpha$ is such that $\mu(\Delta_\alpha) > 0$, then the denominator in the definition of $\kappa_\alpha(\bm{x})$ in \eqref{eq:kappax} is positive for all $\bm{x}$ in a neighbourhood of $\bm{1}_\alpha$.
\end{proof}

\begin{proof}[Proof of Proposition~\ref{thm:multi-peng}]
Proposition~\ref{prop:rnx} implies, as $n \to \infty$, the weak convergence

\[
  \bigl( 
    \sqrt k\{ \widehat r_\alpha(\mb{2}_\alpha) - r_\alpha(\mb{2}_\alpha \},
    \sqrt k\{ \widehat r_\alpha(\mb{1}_\alpha) - r_\alpha(\mb{1}_\alpha \}
  \bigr)
  \wto
  \bigl(
    Z_\alpha(\mb 2_\alpha), Z_\alpha(\mb 1_\alpha)
  \bigr).
\]
Now $\widehat \eta_\alpha^P = g(\widehat r_\alpha(\mb 2_\alpha), \widehat r_\alpha (\mb 1_\alpha))$ and $\eta_\alpha = 1 = g(r_\alpha(\mb 2_\alpha),  r_\alpha (\mb 1_\alpha)) = g(2\rho_\alpha,\rho_\alpha)$, with $g(x,y) = \log(2)/ \log (x/y)$; note that the function $r_\alpha$ is homogeneous. Since the gradient of $g$ is $\nabla g(x,y) = \log(2) (\log(x/y))^{-2} ( - x^{-1}, y^{-1} )$, the delta method gives
\begin{align*}
  \sqrt k (\widehat \eta^P - 1) \; \wto \; 
  &
  \left\langle 
    \nabla g(2\rho_\alpha, \rho_\alpha) ,\, \big( Z_\alpha(\mb 2_\alpha),\, Z_\alpha(\mb 1_\alpha) \big) 
  \right\rangle \\
  &=
  \frac{1}{\rho_\alpha \log 2 } 
  \left\langle ( -1/2, 1   ) , \, \big( Z_\alpha(\mb 2_\alpha), Z_\alpha(\mb 1_\alpha)\big) \right\rangle \\
  &= 
  \frac{-1}{2 \rho_\alpha \log 2 } 
  \{ Z_\alpha(\mb 2_\alpha) - 2 Z_\alpha(\mb 1_\alpha) \}.
\end{align*}
The first part of the assertion follows. As for the variance,
\begin{align*}
  \Var(  Z_\alpha(\mb 2_\alpha) - 2Z_\alpha(\mb 1_\alpha) ) 
  &= 
  \Var( Z_\alpha(\mb 2_\alpha) )
  + 
  4\Var( Z_\alpha(\mb 1_\alpha) )
  -
  4 \Cov(Z_\alpha(\mb 2_\alpha), Z_\alpha(\mb 1_\alpha)),  
\end{align*}
The function $r_\alpha$ is homogeneous of order $1$, so that  $\partial_j{r}_{\alpha}$ is constant along rays, that is, the function $0 < t \mapsto \partial_j {r}_{\alpha}(t \,\mb x)$ is constant. Moreover, the measure $\Lambda$ is homogeneous of order $1$ too. In view of \eqref{eq:covW} and \eqref{eq:Z_alpha(x)}, it follows that $\Var( Z_\alpha(t \bm{x}) ) = t \Var( Z_\alpha( \bm{x} ) )$ for $t > 0$; in particular $\Var( Z_\alpha(\bm{2}_\alpha) = 2 \Var( Z_\alpha(\bm{1}_\alpha )$. Further, $\rho_\alpha = (\mathrm{d} r_\alpha(t, \ldots, t) / \mathrm{d}t)_{t=1} = \sum_{j \in \alpha} \dot{\rho}_{j,\alpha}$ and thus 
\begin{align*}
    \Var(Z_\alpha(\mb 1_\alpha))
    & =  \rho_\alpha - 2\sum_{j\in\alpha}\dot{\rho}_{j,\alpha} \rho_\alpha  +  \sum_{j\in\alpha} \sum_{j'\in\alpha}\dot{\rho}_{j,\alpha}\dot{\rho}_{j',\alpha} \rho_{\{j,j'\}} \\
    & = \rho_\alpha - 2 \rho_\alpha^2 +  \sum_{j\in\alpha} \sum_{j'\in\alpha}\dot{\rho}_{j,\alpha}\dot{\rho}_{j',\alpha} \rho_{\{j,j'\}}.
  \end{align*}
The covariance term is
\begin{multline*}
    \Cov(Z_\alpha(\mb 2_\alpha), Z_\alpha(\mb 1_\alpha))
    =\rho_\alpha 
      - \sum_{j\in\alpha} \dot{\rho}_{j,\alpha} \rho_\alpha
      - \sum_{j\in\alpha}\dot{\rho}_{j,\alpha}
      r_\alpha(\mb 2_\alpha \wedge \indinf_{j}) \\
      + \sum_{j\in\alpha} \sum_{j'\in\alpha}
      \dot{\rho}_{j,\alpha}\dot{\rho}_{j',\alpha}r_{\{j,j'\}}(2,1),
\end{multline*}
with $\mb 2_\alpha \wedge \indinf_{j}$ as explained in the statement of the proposition. Since $\sum_{j\in\alpha} \dot{\rho}_{j,\alpha} = \rho_\alpha$, we can simplify and find
\begin{align*}
  \Var(  Z_\alpha(\mb 2_\alpha) - 2Z_\alpha(\mb 1_\alpha) ) 
  &= 
  6 \Var( Z_\alpha(\bm{1}_\alpha) ) - 4 \Cov(Z_\alpha(\mb 2_\alpha), Z_\alpha(\mb 1_\alpha)) \\
  &= 
  2 \rho_\alpha   - 8 \rho_\alpha^2 +  4\sum_{j\in\alpha} \dot{\rho}_{j,\alpha} r_\alpha(\mb 2_\alpha \wedge \indinf_{j})
  \\
  &\qquad \null 
  + 
  \sum_{j\in\alpha} \sum_{j'\in\alpha} 
  \dot{\rho}_{j,\alpha}\dot{\rho}_{j',\alpha} \big[ 6 \rho_{\{j, j'\}} - 4 r_{\{j, j' \}}(2,1)\big].
\end{align*}
Divide the right-hand side by $(2 \rho_\alpha \log 2)^2$ to obtain \eqref{eq:var-hateta}.
\end{proof}

\begin{proof}[Proof of Proposition~\ref{prop:normality-hill}]
   To alleviate
  notations, $\emptyset \ne \alpha \subset\{1,\ldots, d\}$ is fixed
  and the subscript $\alpha$ is omitted throughout the proof.
  Introduce the tail empirical process $Q_n(t) = \widehat T_{(n - \lfloor kt \rfloor)}$ for $0 < t < n/k$.
  The key is to represent the Hill estimator as a statistical tail functional \citep[Example 3.1]{drees1998evindex} of $Q_n$, i.e., $\widehat \eta^H = \Theta(Q_n)$, where $\Theta$ is the map defined for any measurable function
  $ z: (0,1]\to \rset$ as $\Theta(z) = \int_{0}^1  \log^+ \{ z(t) / z(1) \} \ud t$ when the integral is finite and $\Theta(z) = 0$ otherwise. Let $z_\eta: t\in(0,1]\mapsto t^{-\eta}$ denote the quantile function of a standard Pareto distribution with index $1/\eta$; it holds that $\Theta(z_\eta) = \eta$. The map $\Theta$ is scale invariant, i.e., $\Theta(t z) = \Theta(z), t>0$.
  
  The proof consists of three steps:
  \begin{enumerate}
  \item Introduce a function space $D_{\eta,h}$ allowing
    to control $Q_n(t)$ and $z_\eta(t)$ as $t\to 0$. In this space and up to rescaling, 
    $Q_n - z_\eta$ converges weakly
    to a Gaussian process.
  \item Show that the map $\Theta$ is Hadamard
    differentiable at $z_\eta$ tangentially to some well chosen
    subspace of $D_{\eta,h}$.
  \item Apply the functional delta method to show that $\eta^H = \Theta(Q_n) $ is asymptotically normal and compute its asymptotic variance via the Hadamard derivative of $\Theta$.
  \end{enumerate}

  \paragraph{\mdseries \underline{Step 1.}}
  Let $\epsilon>0$ and $h(t) = t^{1/2+ \epsilon},\, t\in[0,1]$. Then
  $h\in\mathcal{H}$, where
  \[\mathcal{H} = \{z: [0,1]\to \rset  \mid z \text{ continuous, }
    \lim_{t\to 0} z(t) t^{-1/2} (\log\log(1/t))^{1/2} = 0
    \}.\] Introduce the function space
  \[
    D_{\eta, h} = \{z: [0,1]\to \rset \mid \lim_{t\to 0} t^{\eta} h(t)
    z(t) = 0\,;\; t\mapsto t^{\eta}h(t) z(t) \in D[0,1] \},
  \]
  where $D[0,1]$ is the space of c\`adl\`ag functions. Notice that
  $z_\eta \in D_{\eta,h}$. Equip $D_{\eta, h}$ with the seminorm $\|z\|_{\eta,h} = \sup_{t\in(0,1]}|t^\eta h(t) z(t)|$. Let $m = \lceil n q^\leftarrow (k/n)\rceil $, with $\lceil \, \cdot \, \rceil$ the ceil function, so that $k/m\to \rho$;
  for self-consistency of the present paper, the roles of
  $k$ and $m$ are reversed compared to the notation in \cite{draisma2004bivariate}. From \cite[Lemma~6.2]{draisma2004bivariate}, we have, for all $t_0>0$, 
  in the space $D_{\eta,h}$, the weak convergence
  \begin{equation}
    \label{eq:cv-tailQuantile}
    \sqrt{k} \left( \frac{m}{n}Q_n - z_\eta \right) \wto \left( \eta t^{-(\eta +1)} \bar W(t)\right)_{t \in [0,t_0]}
  \end{equation}
  where $\bar W (t) = \tilde W(\mb t_\alpha)$, and $\tilde W$ is defined as in the
  statement of Proposition~\ref{prop:normality-hill}. Indeed, the
  process $\bar W$ in the statement from \cite[Lemmata~6.1 and~6.2]{draisma2004bivariate}
  has same distribution as $W_1(\mb t_\alpha)$ in
  the case $\rho=0$; recall that our $\rho$ is denoted by $l$ in \cite{draisma2004bivariate}. Put $U_{i,j} = 1 - F_j(X_{i,j})$, and let $U_{(1),j} \le \ldots \le U_{(d),j}$ be the order statistics of $U_{1,j}, \ldots, U_{n,j}$. In the case $\rho>0$,
  $\bar W$ equals in distribution $W_{\mathrm{dra}} (\mb t_\alpha)$ where
  $W_{\mathrm{dra}}$ appears in Lemma~6.1 in the cited reference as the
  limit in distribution (for $\alpha= \{1,2\}$), for
  $\mb x\in E_\alpha$, of
  \begin{align*}
    \Delta_{n,k,m} (\mb x) 
    &=  \sqrt{k} \Bigg[  \frac{1}{k}  \sum_{i=1}^n \un\{ \forall j \in \alpha : U_{i,j} \le U_{(\lfloor m x_j \rfloor), j} \} - c(\mb x)\Bigg]  \\
    &= 
    \underbrace{ \sqrt{ \frac{m}{k}} }_{\to \rho^{-1/2}}
    \sqrt{m}
    \Biggl[ 
      \underbrace{\frac{1}{m} \sum_{i=1}^n \un\{ \forall j \in \alpha : U_{i,j} \le U_{(\lfloor m x_j \rfloor),j} \} }_{\text{$r_n(\mb x)$ with $k$ replaced by $m$}} 
      \null -
      r(\mb x)\underbrace{\frac{k}{m\rho}}_{\to 1}
    \Biggr].
  \end{align*} 
  From Proposition~\ref{prop:rnx} and Slutsky's Lemma, we have $\Delta_{n,k,m} \wto \rho^{-1/2} Z_\alpha$ in $\ell^\infty([0, 1]^\alpha)$. Therefore, $W_{\mathrm{dra}} = \rho^{-1/2} Z_\alpha$, as claimed.

  \paragraph{\mdseries \underline{Step 2.}}
  The right-hand side of
  \eqref{eq:cv-tailQuantile} 
  belongs to
  $\mathcal{C}_{h,\eta} = \{ z \in D_{\eta,h} \mid \text{$z$ is continuous}\}$.  To apply the functional delta-method   \cite[Theorem~20.8]{van2000asymptotic}, we must verify that
  the restriction of $\Theta$ to $\bar D_{\eta, h}$ is
  Hadamard-differentiable tangentially to $\mathcal{C}_{\eta,h}$, with
  derivative $\Theta'$, where $\bar D_{\eta, h}$ is a subspace of
  $D_{\eta,h}$ such that $\PP(Q_n \in \bar D_{\eta,h})\to 1$ as
  $n\to \infty$; see the remark following Condition~3 in \cite{drees1998evindex}. Then it will follow
  from the scale invariance of $\Theta$, the identities $\Theta(Q_n) = \widehat \eta^H$ and $\Theta(z_\eta) = \eta$, and the weak convergence in~\eqref{eq:cv-tailQuantile} that
  \begin{equation}
  \label{eq:etaH_W}
    \sqrt{k} \left( \widehat\eta^H- \eta \right)
    =
    \sqrt{k} \left( \Theta(\frac{m}{n}Q_n) - \Theta(z_\eta) \right)
    \wto 
    \Theta'\left[\left( \eta t^{-(\eta +1)} \bar W(t)\right)_{t \in [0,1]}\right]
  \end{equation}
  as $n \to \infty$. From \cite[Example~3.1]{drees1998evindex}, the
  restriction of $\Theta$ to $\bar D_{\eta,h}$, the subset of
  functions on $D_{\eta,h}$ which are positive and non increasing, is
  indeed Hadamard differentiable; letting $\nu$ denote the measure $\ud \nu (t) = t^{\eta}\ud t + \ud \epsilon_1(t)$, with $\epsilon_1$ a point mass at $1$, the derivative is
  \[
    \Theta'(z) = \int_0^1 t^\eta z(t) \ud t - y(1) = \int_{[0,1]} z(t)
    \ud \nu (t).
  \]
  
  \paragraph{\mdseries \underline{Step 3.}}
  The weak limit in \eqref{eq:etaH_W} is thus equal to $\int_{[0,1]} \eta t^{-(\eta +1)} \bar W(t) \ud \nu(t)$.  From~\cite[Proposition 2.2.1]{shorack2009empirical}, the latter random variable is centered Gaussian with variance
  \[
    \sigma^2 = \iint_{[0,1]^2} \eta^2 (st)^{-(\eta +1)} \Cov(\bar
    W(s), \bar W(t))\ud \nu(s)\ud\nu(t).
  \]
  By definition of $\nu$ and by symmetry of the covariance,

  \begin{multline*}
    \sigma^2 / \eta^2  = 2  \underbrace{\int_{s=0}^1\int_{t=0}^s (st)^{-1} \Cov(\bar W(s), \bar W(t)) \ud t \ud s}_{A} \\
                         - 2 \underbrace{\int_{s=0}^1 \Cov(\bar W(s), \bar W(1)) s^{-1}\ud s}_B + \Var(\bar W(1)).  
  \end{multline*}
  For any $s\in(0,1)$,
  \begin{align*}
    \int_{t=0}^s \Cov(\bar W(s), \bar W(t))(st)^{-1} \ud t
    & =\int_{u=0}^1 \Cov(\bar W(s), \bar W(us))(su)^{-1} \ud u \\
    &= \int_{u=0}^1 \Cov(\bar W(1), \bar W(u))(u)^{-1} \ud u  
    = B. 
  \end{align*}
  The penultimate equality follows from $\Cov(\bar W(\lambda s),\bar W( \lambda t)) = \lambda \Cov(\bar W(s),\bar W( t))$
  for   $\lambda>0$ and $s,t\in(0,1]$. Therefore $A=B$ and
  $\sigma^2 = \eta^2 \Var(\bar W(1))$, as required.
\end{proof}


\section{CLEF algorithm and variants}
\label{sec:appendix-CLEF}

The CLEF algorithm is described at length in \cite{chiapinofeature}. For completeness, its pseudo-code is provided below. The underlying idea is to iteratively construct pairs, triplets, quadruplets\dots\ of features that are declared `dependent' whenever $\widehat\kappa_\alpha \ge C$ for some user-defined tolerance level $C > 0$. Varying this criterion produces three variants of the original algorithm, namely CLEF-Asymptotic, CLEF-Peng, and CLEF-Hill. The pruning stage of the algorithm is the same for all three variants.
\begin{algorithm}
\caption{CLEF  (CLustering Extreme Features) }
\begin{algorithmic}
  \STATE \textbf{Input}: Tolerance parameter $\kappa_{\min}>0$. \\[.7ex]
  \STATE {\bf STAGE 1: constructing the collection $\widehat{\mathbb{M}}$ of tail-dependent groups.}
  \STATE \textbf{Step 1:} Put $\hat{\mathcal{A}}_1 = \{ \{1\}, \ldots, \{d\} \}$ and $S = 1$.
  \STATE \textbf{Step $\bm{s = 2, \ldots, d}$}: If $\hat{\mathcal{A}}_{s-1}=\emptyset$, end \textbf{STAGE 1}. Otherwise:
\begin{itemize}
\item  Generate  candidates of size $s$: \\
 $\mathcal{A}'_{s} = \{\alpha\subset\{1,\ldots,d\} : |\alpha|=s \text{ and }
\alpha\setminus j \in\hat{\mathcal{A}}_{s-1} \text{ for all } j\in\alpha\}$.
\item Put  $ \hat{\mathcal{A}}_{s}
= \big\{\alpha\in\mathcal{A}'_{s} : \hat{\kappa}_{\alpha}>\kappa_{\min} \big\}$.
\item If $\hat{\mathcal{A}}_{s} \ne \emptyset$, put $S = s$.
\end{itemize}
\emph{\bf Output}: $\widehat{\mathbb{M}} = \emptyset$ if $S = 1$ and $\widehat{\mathbb{M}} =
\bigcup_{s=2}^{S} \hat{\mathcal{A}}_s$ if $S \ge 2$. \\[.7ex]
\STATE {\bf STAGE 2: pruning, keeping maximal groups $\bm\alpha$ only.}\\
\STATE If $S = 1$, then $\widehat{\mathbb{M}}_{\max} = \emptyset$. Otherwise:
\STATE \textit{Initialization:} $\mathbb{\widehat M}_{\max} \leftarrow \hat{\mathcal{A}}_S$. 
\STATE  for $s= (S-1) : 2$, 
\STATE \qquad for $\alpha\in\hat{\mathcal{A}}_s$, 
\STATE \qquad \qquad If  there is no
       $\beta\in\widehat{\mathbb{M}}_{\max}$ such that
       $\alpha\subset\beta$, then $\widehat{\mathbb{M}}_{\max} \leftarrow
\widehat{\mathbb{M}}_{\max} \cup\{\alpha \}$.
\STATE {\bf Output}: $\widehat{\mathbb{M}}_{\max} $
\end{algorithmic}
\end{algorithm}

\FloatBarrier

\begin{acknowledgements}
  This work was supported by a public grant as part of the
  Investissement d’avenir project, reference ANR-11-LABX-0056-LMH,
  LabEx LMH.
\end{acknowledgements}


\end{document}